\documentclass[oribibl]{llncs}
\usepackage{nico}
\usepackage{permissive}
\usepackage{microtype}

\renewcommand{\orcidID}[1]
{\unskip\texorpdfstring{%
\href{https://orcid.org/#1}{\protect\includegraphics[width=10px]{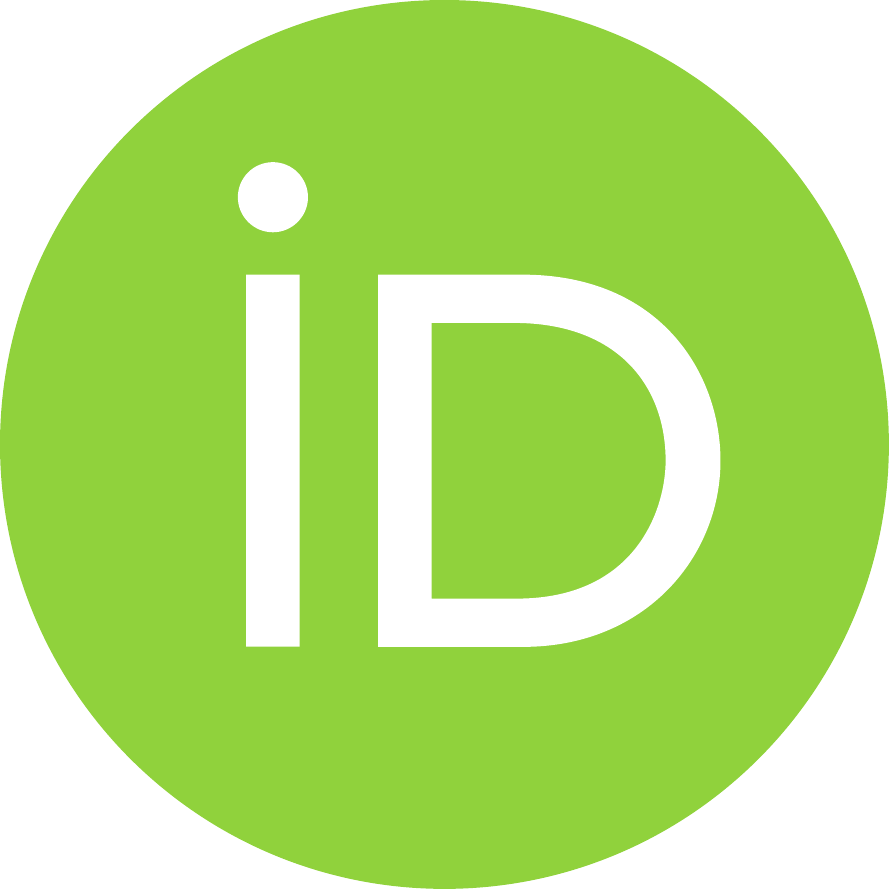}}}{}}
\title{Computing maximally-permissive strategies\\
  in acyclic timed automata\thanks{This work was partially funded by ANR project Ticktac (ANR-18-CE40-0015).}}
\author{Emily~Clement\inst{1,2} \orcidID{0000-0002-8105-665X} \and
  Thierry~J\'eron\inst{1}
  \orcidID{0000-0002-9922-6186}
  \and
  Nicolas~Markey\inst{1}
  \orcidID{0000-0003-1977-7525}
  \and
  David~Mentr\'e\inst{2}
  \orcidID{0000-0003-4315-0335}}
\institute{\leavevmode
  \inst{1} IRISA, Inria \& CNRS \& Univ. Rennes, France \\
  \email{firstname.lastname@inria.fr} \\\smallskip
  \inst{2} Mitsubishi Electric R\&D Centre Europe -- Rennes, France\\
  \email{initial-of-firstname.lastname@fr.merce.mee.com}}

\begin{document}
\maketitle

\begin{abstract}
Timed automata are a convenient mathematical model for modelling and
reasoning about real-time systems. While they provide a powerful way
of representing timing aspects of such systems, timed automata assume
arbitrary precision and zero-delay actions; in~particular, a~state
might be declared reachable in a timed automaton, but impossible to
reach in the physical system it models.

In this paper, we~consider \emph{permissive} strategies as a way to
overcome this problem: such strategies propose \emph{intervals of
  delays} instead of single delays, and aim at reaching a target state
whichever delay actually takes place. We~develop an algorithm for
computing the optimal permissiveness (and~an~associated
maximally-permissive strategy) in~acyclic timed automata and games.

\end{abstract}

\section{Introduction}
\label{sec-intro}

Timed automata~\cite{AD94} are a powerful formalism for modelling and
reasoning about real-time computer systems: they offer a convenient
way of modelling timing conditions (not relying on discretization)
while allowing for efficient verification algorithms; as~a
consequence, they have been widely studied by the formal-verification
community, and have been applied to numerous industrial case studies
thanks to advanced tools such as Uppaal~\cite{BDL+06}, TChecker~\cite{HPT19} or
Chronos~\cite{BDM+98}.

One drawback of timed automata is that they are a mathematical model,
assuming infinite precision in the measure of time; this does not
correspond to physical devices such as computers. As~a consequence,
properties that are proven to hold on the model may fail to hold on
any implementation. As~a very simple example, consider
two (or even infinitely-many) consecutive actions that
have to be performed at the exact same time: while this would be
possible in a mathematical model, this would not be possible on a
physical device.

Several approaches have attempted to address such problems, depending
on the property to be checked. When considering safety properties,
timing imprecisions may add new behaviours, which have to be taken
into account in the safety check. In~that setting, \emph{guard
enlargement}~\cite{Pur00,DDMR04} has been proposed as a way to model the
fact that some timing conditions might be considered true even if they
are (slightly) violated: the existence of an enlargement value for
which the set of executions is safe is decidable.  When dealing with
reachability properties, timing imprecisions may prevent a run to be
valid. A~topological approach has been proposed, where a state is
declared reachable only if there is a \emph{tube of trajectories}
reaching the target state~\cite{GHJ97}.
Game-based approaches have also been proposed, where a state is said
reachable if there is a strategy to reach this state when an opponent
player is allowed to modify (up~to a certain point) the values of the
delays~\cite{BMS15,BFM15}.

In this paper, we build on the approach of~\cite{BFM15}, where the
authors aim at computing \emph{maximally-permissive} strategies for
reaching a target state. While in classical timed automata,
reachability is witnessed by a sequence of delays and transitions
leading to a target state, here the aim is to propose \emph{intervals}
of delays, leaving it to an opponent player to decide which delay will
indeed take place. Of~course, the strategy has to be able to respond
to any choice of the opponent, eventually reaching the target
state.

We~can then have several ways of measuring permissiveness of a
strategy, the general idea being that larger intervals of delays are
preferred. In~\cite{BFM15}, each interval is associated with a penalty,
which is the inverse of the length of the interval. Penalties are
summed up along paths, and maximally-permissive strategies are those
having minimal worst-case penalty. This favours both large
intervals and short paths, but computing optimal strategies could only
be achieved in the case of one-clock timed automata in~\cite{BFM15}.

In the present paper, permissiveness of a strategy is defined as the
size of the smallest interval proposed by that strategy. We~develop an
algorithm to compute the permissiveness of any (winning) configuration
in acyclic timed automata and games, with any number of clocks.
Consider for instance a scheduling problem, where a number of tasks
have to be performed in a certain order within a given
delay. Classical reachability algorithms would just say whether a
given set of tasks are schedulable (in~the mathematical model); this
then requires launching some of the tasks at very precise dates, as
the computed schedule need not be correct if delays are slightly
modified.  Instead, our~algorithm could compute the permissiveness of
the best schedule, thereby measuring the amount of imprecision that
can be allowed, depending on the deadline by which all tasks have to
be finished.

This paper is organized as follows: in~Section~\ref{sec-defs},
we~introduce the necessary definitions, in particular of timed
automata and permissiveness of strategies, and prove basic
results. Section~\ref{sec-linear} is devoted to solving the case
of \emph{linear} timed automata, where all states have at most one
outgoing transition, thereby focusing only on choices of
delays. Section~\ref{sec-generic} extends this to acyclic
timed automata and games.

By lack of space, most of the proofs could not be included in this
version of the paper.
They can be found in the long version~\cite{arxiv-2007.01815} of this article.






\section{Definitions}
\label{sec-defs}

%

%

\subsection{Piecewise-affine functions} 

A~valuation for a set~$\clocks$ of variables is a mapping $v\colon
\clocks \to \bbR+$, assigning a nonnegative real value to each
variable. We~write~$\vzero$ for the valuation defined as $\vzero(c)=0$
for any~$c\in\clocks$. We~write $(\bbR+)^{\clocks}$ for the set of
valuations for~$\clocks$, which we~identify with $(\bbR+)^n$ when
$\clocks$ has exactly~$n$ variables.
We~write $\bbRbar$ for $\bbR\cup\{-\infty;+\infty\}$.

\begin{definition}
  An $n$-dimensional affine function is a mapping
  $f\colon \bbR+^n \to \bbRbar$ s.t.%
  \begin{itemize}
  \item either there exists a vector $(F_k)_{0\leq k\leq n}\in
    \bbR^{n+1}$ such that $f(v)=F_0+\sum_{1\leq i\leq n} F_i\cdot
    v_i$;
  \item or $f(v)=-\infty$ for all~$v\in\bbR+^n$; in that case we can
    still write $f(v)=F_0+\sum_{1\leq i\leq n} F_i\cdot v_i$, by
    setting $F_0=-\infty$ and $F_i=0$ for all~$1\leq i\leq n$;
  \item or $f(v)=+\infty$ for all~$v\in\bbR+^n$; similarly, this
    corresponds to setting $F_0=+\infty$ and $F_i=0$ for all~$1\leq
    i\leq n$.
  \end{itemize}
  A linear function~$f$ is an affine function for which~$f(\vzero)=0$.
\end{definition}


If $\Phi=(\phi_k)_{1\leq k\leq m}$ is a set of $n$-dimensional affine
functions and $b=(b_k)_{1\leq k\leq m}$ is a set of intervals,
we~write $\polyh{\Phi,b}$ for the intersection $\bigcap_{1\leq k\leq m}
\phi_k^{-1}(b_k)$. This defines a convex polyhedron of~$\bbR+^n$.

An $n$-dimensional piecewise-affine function is a mapping $f\colon
\bbR+^n \to \bbRbar$ for which there exists a
partition~$S=(S_k)_{1\leq k\leq m}$ of~$\bbR+^n$ into convex
polyhedra, and a family~$(f_k)_{1\leq k\leq m}$ of affine functions
such that for any~$x\in\bbR+^n$, writing~$k$ for the (unique) index
in~$[1;m]$ such that $x\in S_k$, it~holds $f(x)=f^k(x)$.

\subsection{Timed automata}

Given a valuation~$v$ and a nonnegative real~$d$,
we~denote with $v+d$ the valuation~$w$ such that
$w(c)=v(c)+d$ for all~$c\in\clocks$. For any
subset~$I\subseteq \bbR+$, we~write $v+I$ for the set of valuations
$\{v+d \mid d\in I\}$.
Given a valuation~$v$ and a subset~$r\subseteq\clocks$, we~write
$v[r\to 0]$ for the valuation~$w$ such that
$w(c)=0$ if $c\in r$ and $w(c)=v(c)$ if $c\notin r$.

The~set of linear constraints over~$\clocks$ is defined as
\(
\constr(\clocks) \ni g\coloncolonequals c \sim n \mid g\et g
\)
where $c$ ranges over~$\clocks$, $n$~ranges over~\bbN, and
$\mathord\sim\in\{\mathord<,\mathord\leq,\mathord=,\mathord\geq,
\mathord>\}$. That a~clock valuation~$v$ satisfies a clock
constraint~$g$, denoted $v\models g$ (and sometimes $v\in g$,
seeing~$g$ as a convex polyhedron), is defined inductively~as
\begin{xalignat*}2
  v\models c\sim n & \iff v(c) \sim n
  & v\models g_1\et g_2  & \iff v\models g_1 \text{ and } v\models g_2
\end{xalignat*}

For the rest of this paper, we~fix a finite alphabet~$\Alp$.
\begin{definition}[\cite{AD94}]
A \newdef{timed automaton} over~$\Alp$ is a tuple
$\calA=\tuple{\clocks,L,T,I}$ where $\clocks$ is a finite set of clocks,
$L$ is a finite set of states (or \newdef{locations}), and $T\subseteq
L\times \constr(\clocks)\times \Alp\times 2^{\clocks}\times L$ is a
finite set of transitions, and $I\colon S\to \constr(\clocks)$
define the invariant constraints in locations.
%
\end{definition}

A~\newdef{configuration} of a timed automaton is a pair~$(\loc,v)$ where
$\loc$~is a location of the automaton and $v$~is a clock valuation
such that $v\models I(\loc)$.
The~semantics of timed automata can be defined as an infinite-state
labelled transition system whose states are the set of configurations,
and whose transitions are of two kinds:
\begin{itemize}
\item \newdef{delay transitions} model time elapsing: no transitions
  of the timed automaton are taken, but the values of all clocks are
  augmented by the same value. For any configuration~$(\loc,v)$ and
  any delay $d\in \bbR+$, there is a transition $(\loc,v) \trans[d]
  (\loc,v+d)$, provided that $v+d\models I(\loc)$;
\item \newdef{action transitions} represent the effect of taking a
  transition in the timed automaton. For~any configuration~$(\loc,v)$
  and any transition~$t=(\loc,g,a,r,\loc')$, if $v\models g$, then
  there is a transition $(\loc,v) \trans[a] (\loc',v[r\to 0])$,
  provided that $v[r\to 0]\models I(\loc')$.
\end{itemize}
We~write $(\loc,v)\trans[d,a](\loc',v')$ when there exists $(\loc'',v'')$ such
that $(\loc,v)\trans[d](\loc'',v'')$ and $(\loc'',v'')\trans[a](\loc',v')$.
A~run of a timed automaton is a sequence of
configurations~$(\loc_i,v_i)_i$ such that there exists $d\in\bbR+$ and
$a\in\Alp$ such that $(\loc_i,v_i)\trans[d,a] (\loc_{i+1},v_{i+1})$ for
all~$i$.  Even if it means adding a sink state and corresponding
transitions, we assume that from any configuration, there always
exists a transition $\trans[d,a]$ for some~$d\in\bbR+$ and some
$a\in\Alp$. This way, any finite run can be extended into an infinite
run (in~terms of its number of transitions). We~also assume that, from any
location~$\loc$ and any action~$a$,
there is at most one transition from~$\loc$ labelled with~$a$.

One of the most basic problems concerning timed automata is that of
reachability of a location: given a timed automaton~$\Aut$, a~source
configuration~$(\loc_0,v_0)$ (usually assuming~$v_0=\vzero)$ and a
target location~$\loc_f$, it~amounts to deciding whether there exists
a run from~$(\loc_0,v_0)$ to some configuration~$(\loc_f,v_f)$ in the
infinite-state transition system defining the semantics
of~$\Aut$. This problem has been proven decidable
(and \PSPACE-complete) in the early 1990s~\cite{AD94}, using
\newdef{region equivalence}, which provides a finite-state automaton
that is (time-abstracted) bisimilar to the original timed automaton.


\subsection{Permissive strategies in timed automata}
Solving reachability using the algorithm above, we~can obtain a
sequence of delays and transitions to be taken for reaching the target
location. Playing this sequence of delays and transitions however
requires infinite precision in order to meet all timing constraints,
which might not be possible on physical devices.


In this paper, we~address this problem by building on the setting
studied in~\cite{BFM15}: in~that setting, the delays that are played
may be slightly perturbed, and it~can be required to adapt the future
delays (and possibly actions) so as to make sure that the target is
indeed reached. 


We~encode the imprecisions using a game setting: the player proposes
an interval of possible delays (together with the action to be
played), and its opponent selects, in the proposed interval,
the exact delay that will take place. 

Formally, in our setting, a \emph{move} from some
configuration~$(\loc,v)$ is a pair~$(I,a)$, where $I\subseteq \bbR+$
is a closed\footnote{We~only consider closed intervals here to
  simplify the presentation.}
interval, possibly right-unbounded,
and $a\in\Sigma$, such that there is a
transition~$(\loc,g,a,r,\loc')$ for which $v+I\subseteq g$ (\emph{i.e.}, for
any valuation $w\in v+I$, it~holds $w\models g$).  We~write
$\moves(\loc,v)$ for the set of moves from~$(\loc,v)$.

A~\newdef{permissive strategy} is a function~$\sigma$
mapping finite runs $(\loc_i,v_i)_{0\leq i\leq n}$ to moves in
$\moves(\loc_n,v_n)$.  A~run~$\rho=(\loc_i,v_i)_i$ is
\newdef{compatible} with a permissive strategy~$\sigma$ if, for any
finite prefix $\pi=(\loc_i,v_i)_{0\leq i\leq j}$ of~$\rho$,
$\sigma(\pi)$ is defined and, writing $\sigma(\pi)=(I,a)$, there
exists $d\in I$ such that
$(\loc_j,v_j)\trans[d,a](\loc_{j+1},v_{j+1})$.  A~permissive
strategy~$\sigma$ is winning from a given configuration~$(\loc_0,v_0)$
if any infinite run originating from~$(\loc_0,v_0)$ that is compatible
with~$\sigma$ is winning (which, in our setting, means that it visits
the target location~$\loc_f$).  Notice that classical strategies
(which propose single delays instead of intervals of delays) are
special cases of permissive strategies. It~follows that, as~soon
as there is a path
from some configuration~$(\loc,v)$ to~$\loc_f$, there exists a
winning permissive strategy from~$(\loc,v)$ (possibly proposing punctual
intervals). Such configurations are said winning, and the~\emph{winning
  zone} is the set of all winning configurations.

Our aim is to compute \emph{maximally-permissive} winning
strategies. In~this~work, we~measure the permissiveness of a strategy~$\sigma$
in a configuration~$(\loc,v)$, denoted $\Perm_\sigma(\loc,v)$, as
the length of the smallest interval it may return.
Formally:
\begin{definition}
Let~$\sigma$ be a permissive strategy, and $(\loc,v)$ be a
configuration of~$\calA$. The~permissiveness of~$\sigma$ in~$(\loc,v)$,
denoted with~$\Perm_\sigma(\loc,v)$, 
is defined as follows:%
\begin{itemize}
\item if~$\sigma$ is not winning from~$(\loc,v)$,
  the~permissiveness of~$\sigma$ in~$(\loc,v)$ is~$-\infty$;
\item otherwise, $\Perm_\sigma(\loc,v) = \inf\{|I| \mid
  \exists \pi.\ \sigma(\pi)=(I,a) \text{ for some~$a$}\}$. 
\end{itemize}
The~permissiveness of configuration~$(\loc,v)$ is then defined as
\[
\Perm(\loc,v)=\sup_{\sigma} \Perm_{\sigma}(\loc,v).
\]
\end{definition}
In this paper, we prove that $\Perm$~is a piecewise affine function,
and develop an algorithm for computing that function.
Intuitively, this corresponds to computing how much precision is
needed in order to reach the target configuration.


%
%

\begin{remark}\label{remark-related-work} 	
Notice that our definition of permissiveness is similar in spirit
with that of~\cite{BFM15}. However, in~\cite{BFM15}, each
move~$(I,a)$ was associated a penalty (namely~$1/|I|$), and
penalties are summed up along the execution. This tends to make the
player favour shorter paths with possibly small intervals (hence
demanding more accuracy when playing) over long paths with larger
intervals. Our setting only aims at maximizing the size of the
smallest interval to be played.

Our work can also be seen as a kind of quantitative extension of
\emph{tubes of trajectories} of~\cite{GHJ97}: permissiveness could
be seen as the minimal width of such a tube. However, we~are in a
game-based setting, and (except in Section~\ref{sec-linear}) the
strategy could suggest to take different transitions if they allow
for more permissiveness.

Finally, and perhaps more importantly, our setting is quite close to
that of~\cite{BMS15}, but with a more quantitative focus: we~aim at
computing the optimal permissiveness for all winning configurations
(with reachability objective), while only a global lower bound (of
the form~$1/m$ where $m$~is doubly-exponential in the size of the
input) is obtained in~\cite{BMS15}.  Similar results to those
of~\cite{BMS15} are obtained in~\cite{SBMR13, BMRS19} for B\"uchi
objectives; such extensions are part of our future work.

%
%
%
\end{remark}

\begin{remark}
  The term \emph{permissive} strategy is sometimes used to refer to
  \emph{non-deterministic}, returning the set of all moves that lead
  to winning configurations. In~particular,
  Uppaal-Tiga~\cite{BCDFLL07} can compute maximally-permissive
  strategies in that sense. But~this is only a local view of
  permissiveness, while our aim is to allow for high permissiveness
  all along the execution.
\end{remark}

\subsection{Iterative computation of permissiveness}

Towards computing $\Perm$, we define:
\begin{xalignat*}2
  \calP_i(\loc_f,v)&= +\infty &&
    \qquad\text{ for all valuations~$v$ and all~$i\geq 0$;} \\
 \calP_0(\loc,v)&= -\infty &&
    \qquad\text{ for all valuations~$v$, and for
  all~$\loc\not=\loc_f$;} \\
\calP_{i+1}(\loc,v) &\multicolumn{3}{l}{${}=\left\{\begin{array}{lr} 
   \multicolumn{2}{l}{\sup\limits_{\!\!(I,a)\in\moves(\loc,v)\!\!}
   \min(|I|, 
   \inf \{\calP_{i}(\loc',v')\mid
   \exists d\in I.\ (\loc,v)\trans[d,a](\loc',v')\})} \\[-2mm]
  & \text{ if } \moves( \loc , v ) \not= \varnothing \\[2mm]
  - \infty & \text{otherwise} 
  \end{array}\right.$}
\end{xalignat*}

In the rest of section, we prove some basic properties of this sequence of
functions, and in particular its link with permissiveness. The next
sections will be devoted to its computation on acyclic timed automata.

\medskip
Our first two results are concerned with the evolution of the sequence
with~$i$. They are proved by straightforward inductions.
\begin{restatable}{lemma}{lemmaincri}
\label{lemma-incri}
	For any $(\loc,v)$, the sequence $ \left(
        \calP_i(\loc,v)\right)_{ i \in \mathbb{N} }$ is nondecreasing.
\end{restatable}
\begin{restatable}{lemma}{lemmaconverge}
\label{lemma-converge}
  If the longest path from~$\loc$ to~$\loc_f$ has at most $i$ transitions,
  then for any~$v$ and any~$j\geq 0$,
  it~holds $\calP_{i+j}(\loc,v)=\calP_i(\loc,v)$.
\end{restatable}
%

\medskip
\noindent
The following lemma ties the link between the sequence~$(\calP_i)$ and
permissiveness:%
%
%
\begin{restatable}{proposition}{propinfty}
	\label{prop-infty}
  For any~$i\in\bbN$ and for any configuration~$(\loc,v)$, it~holds:
  \begin{enumerate}
  \item $\calP_i(\loc,v)=-\infty$ if, and only~if, there are no runs
    of length at~most~$i$ from~$(\loc,v)$ to~$\loc_f$;
  \item for any~$p\in\bbR+$, and any~$i\in\bbN$,
    it~holds $\calP_i(\loc,v)>p$
    if, and only~if, 
    there is a permissive strategy with
    permissiveness larger than $p$ that is winning from~$(\loc,v)$ within
    $i$~steps.
  \end{enumerate}

\end{restatable}

\begin{proof}
 We~begin with the first equivalence, which we prove by induction
 on~$i$.  The result is trivial for~$i=0$. Now, assume that the
 result holds up to index~$i$. There may be two reasons for having
 $\calP_{i+1}(\loc,v)=-\infty$ for some~$(\loc,v)$: either
 $\moves(\loc,v)$ is empty, or it is not empty and for any
 $(I,a)\in\moves(\loc,v)$, it~holds 
 \[
 \inf \{\calP_i(\loc',v') \mid \exists d\in I. (\loc,v)\xrightarrow{d,a}
   (\loc',v') \}=-\infty.
 \]
 This is true in particular when~$I=\{d\}$ is punctual: for
 any~$(d,a)$, the successor~$(\ell',v')$ such that
 $(\loc,v)\xrightarrow{d,a} (\loc',v')$ is such that
 $\calP_i(\loc',v')=-\infty$. From the induction hypothesis, there
 can be no path from those~$(\loc',v')$ to~$\loc_f$ with $i$ steps
 or~less. Hence there are no paths from~$(\loc,v)$ to~$\loc_f$ with
 at most~$i+1$ steps.

 Conversely, if there are no paths having at most $i+1$ steps
 from~$(\loc,v)$ to~$\loc_f$, then either this is because
 $\moves(\loc,v)=\varnothing$, or this is because all moves lead to a
 configuration from which there are no paths of length at
 most~$i$ to~$\loc_f$.
 By~induction hypothesis, all successor configurations have
 infinite~$\calP_i$, hence also ${\calP_{i+1}(\loc,v)=-\infty}$.

 \medskip

 We now prove the second claim, still by induction. The~base case is
 again trivial. Now, assume that the result holds up to some
 index~$i$.  We~fix some~$p\in\bbR+$, and first consider a
 configuration~$(\loc,v)$ with
 $\calP_{i+1}(\loc,v)>p$. This~entails that $\moves(\loc,v)$ is
 non-empty, and that 
 there is a move~$(I,a)$ with $|I|>p$
 such that $\calP_i(\loc',v')> p$ for all $(\loc',v')$ such that
 $(\loc,v)\xrightarrow{d,a} (\loc',v')$ with~$d\in
 I$. Applying the induction hypothesis,
 there is an $i$-step winning strategy with permissiveness
 larger than~$p$ from each successor configuration~$(\loc',v')$, from which
 we can build an $i+1$-step winning strategy with permissiveness larger
 than~$p$ from~$(\loc,v)$.

 Conversely, pick an $i+1$-step winning
 strategy~$\sigma_{p}$ from~$(\loc,v)$ with permissiveness
 larger than~$p$. Write~$\sigma_p(\loc,v)=(I_0,a_0)$.
 Then for any~$d\in I_0$, in the location~$(\loc',v')$ such that
 $(\loc,v)\xrightarrow{d,a_0} (\loc',v')$, strategy~$\sigma_p$ is an
 $i$-step winning strategy with permissiveness larger than~$p$, so that,
 following the induction hypothesis, 
 $\calP_i(\loc,v) > p$. It~immediately follows that
 $\calP_{i+1}(\loc,v) >p$. 
\end{proof}




\medskip
Our next three results focus on properties of the functions~$\calP_i$.
First, we~identify zones on which $\calP_i$ is constant. This will be
useful for proving correctness of our algorithm computing~$\calP_i$ in
the next section:
\begin{restatable}{lemma}{lemmainfini}
\label{lemma-infini}
  Let~$\calA$ be a timed automaton, with maximal constant~$M$.
  Let~$\loc$ be a location, and~$i\in\bbN$. Take
  two valuations~$v$ and~$v'$ such that, for any clock~$c$, we have
  either $v(c)=v'(c)$, or $v(c)>M$ and $v'(c)>M$.  
  Then $\calP_i(\loc,v)=\calP_i(\loc,v')$.
\end{restatable}


%
Next we prove that the functions $\calP_i$ are $2$-Lipschitz
continuous (on the zone where they take finite values):
\begin{restatable}{proposition}{propcontinuite}
	\label{prop-continuite}
	For any integer~$i\in\bbN$ and any location~$\loc$, the function
	$\tau_\loc\colon v\mapsto \calP_i(\loc,v)$ is $2$-Lipschitz on the set of
	valuations where it takes finite values.
\end{restatable}

Finally, the following lemma shows the (rather obvious) fact
that $\calP_i(\loc,v+t) \leq \calP_i(\loc,v)$. 
A~consequence of this property is that, for any
non-resetting transition, the~optimal choice for the opponent is the
largest delay in the interval proposed by the player\footnote{This also holds
for any transition in a one-clock timed automaton (because in case the
clock is reset, the new valuation does not depend on the delay chosen
by the opponent).}.
This~corresponds to the intuition that by playing
later, the opponent will force the player to react faster at the next
step. As~Example~\ref{ex-perm} below shows, this~is not the case in
general: in~that example,
from~$(\loc_0, \langle x=0;y=0\rangle)$,
if~the player proposes interval~$[1/4;1]$, the~optimal choice for the
opponent is~$d=1/4$).

\begin{restatable}{lemma}{lemmadecrdiag}
\label{lemma-decrdiag}
  Let $(\ell,v)$ be a configuration, $t\in\bbR+$ such that
  $(\ell,v+t)$ is a configuration of the automaton, and $i\in\bbN$.
  Then $\calP_i(\ell,v)-t \leq \calP_i(\ell,v+t) \leq
  \calP_i(\ell,v)$.

\end{restatable}
%
%



\begin{example}\label{ex-perm}
  Consider the automaton of Fig.~\ref{fig-experm}. We~compute the optimal
permissiveness (and~corresponding strategies) for this small
example. First, $\calP_i(\loc_f,v)=+\infty$ for all~$i$, and
$\calP_0(\loc_0,v)=\calP_0(\loc_1,v)=-\infty$.

\begin{figure}[!ht]
  \centering
  \begin{tikzpicture}
    \begin{scope}
      \draw (0,0) node[rond,bleu] (a) {$\loc_0$};
      \draw (4,0) node[rond,rouge] (b) {$\loc_1$};
      \draw (8,0) node[rond,vert] (c) {$\loc_f$};
      \draw (a) edge[draw,-latex'] node[above] {$\begin{array}{>{\scriptstyle}c}
          0\leq x\leq 1 \\ 0\leq y\leq 1
        \end{array}$} node[below] {$\scriptstyle y:=0$} (b);
      \draw (b) edge[draw,-latex'] node[above] {$\begin{array}{>{\scriptstyle}c}
          1\leq x\leq 2 \\ 0\leq y\leq 1
        \end{array}$} (c);
    \end{scope}
    \begin{scope}[xshift=-1.6cm,yshift=-3.6cm,scale=2.3]
      \fill[fill=blue!40!white] (0,0) -| (.5,.5) -- cycle;
      \shade[left color=blue!40!white, right color=blue!15!white]
        (.5,0) -- (.5,.5) -- (1,1) |- cycle;
      \shade[shading=axis,bottom color=blue!45!white,top color=blue!20!white,shading angle=45]
      (0,0) -- (.5,.5) -- (0,1) -- cycle;
      \shade[bottom color=blue!40!white,top color=blue!15!white]
      (0,1) -- (.5,.5) -- (1,1) -- cycle;

      \foreach \x in {.6,.7,.8,.9} 
        {\draw[blue!70!green!40!white] (\x,0) -- (\x,\x) -- (1-\x,\x) -- (0,2*\x-1);}

      \draw (.37,.15) node {$\frac12$};
      \path (.75,.3) node {$\scriptstyle 1-x$};
      \draw (.5,.85) node {$\scriptstyle 1-y$};
      \draw (.22,.5) node {$\frac{1+x-y}2$};
      \draw (.6,1.1) node {$\scriptstyle -\infty$};
      \draw (1.2,.5) node {$\scriptstyle -\infty$};

      \draw[latex'-latex'] (1.2,0) node[right] {$x$} -| (0,1.2) node[left] {$y$};
      \foreach \x in {0,1}
           {\draw (\x,0) -- +(0,-.05) node[below] {$\scriptstyle \x$};
            \draw (0,\x) -- +(-.05,0) node[left] {$\scriptstyle \x$};
            \draw[dotted] (\x,0) -- +(0,1.2); \draw[dotted] (0,\x) -- +(1.2,0);
           }
      \draw (0,0) -- (1,1); \draw (.5,0) -- (.5,.5) -- (0,1);
      \draw (1,0) |- (0,1);

    \end{scope}
    \begin{scope}[xshift=3.6cm,yshift=-3.6cm,scale=2.3]
      \shade[shading=axis, top color=red!0!white, bottom color=red!35!white, shading angle=45]
      (0,0) -| (1,1) -- cycle;
      \shade[shading=axis, bottom color=red!40!white, top color=red!15!white]
      (1,0) |- (2,1) -- cycle;
      \shade[shading=axis, left color=red!40!white, right color=red!15!white]
      (1,0) -| (2,1) -- cycle;
      \draw[latex'-latex'] (2.2,0) node[right] {$x$} -| (0,1.2) node[left] {$y$};

      \foreach \x in {.2,.4,.6,.8} 
        {\draw[red!40!yellow] (\x,0) -- (1,1-\x) -- (2-\x,1-\x) -- (2-\x,0);}

      \path (.65,.3) node {$\scriptstyle x-y$};
      \path (1.75,.3) node {$\scriptstyle 2-x$};
      \path (1.3,.7) node {$\scriptstyle 1-y$};
      \draw (.6,1.1) node {$\scriptstyle -\infty$};
      \draw (2.2,.5) node {$\scriptstyle -\infty$};
      \draw (.3,.7) node {$\scriptstyle -\infty$};

          \foreach \x in {0,1}
           {\draw (\x,0) -- +(0,-.05) node[below] {$\scriptstyle \x$};
            \draw (0,\x) -- +(-.05,0) node[left] {$\scriptstyle \x$};
            \draw[dotted] (\x,0) -- +(0,1.2);
            \draw[dotted] (0,\x) -- +(2.2,0);
           }
           \draw (2,0) -- +(0,-.05) node[below] {$\scriptstyle 2$};
           \draw[dotted] (2,0) -- +(0,1.2);
      \draw (0,0) -- (1,1) -- (2,1) -- (1,0) -- cycle; 
      \draw (1,1) -- (1,0) -| (2,1);
\end{scope}
  \end{tikzpicture}
  \caption{A linear timed automaton and its permissiveness at~$\loc_0$ and~$\loc_1$}
  \label{fig-experm}
\end{figure}
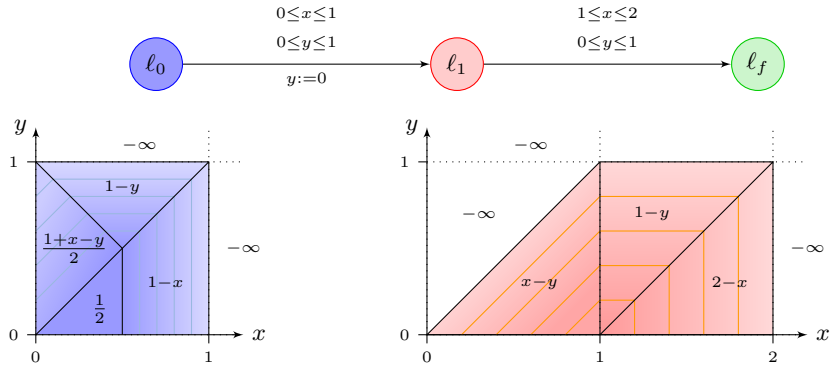

We first focus on~$\loc_1$, with some valuation~$v$: obviously,
if~$v(x)>2$ or~$v(y)>1$, the set~$\moves(\loc_1,v)$ is empty, and
$\calP_1(\loc_1,v)=-\infty$ in that case; similarly if $v(y)>v(x)$.
Since $\calP_0(\loc_f,v)$ does not depend on~$v$, the~optimal move for
the player is the largest possible interval satisfying the
guard:
\begin{itemize}
\item if $v(x)\leq 1$ and $v(y)\leq 1$ (and $v(x)\leq v(y)$), the
  optimal interval of delays is $[1-v(x); 1-v(y)]$, whose length is
  $v(x)-v(y)$;
\item if $v(y)\leq 1\leq v(x)\leq 2$, the transition is immediately
  available, so that the lower bound of the interval will be~$0$. For
  the upper bound, there are two cases:
  \begin{itemize}
  \item if $v(y)\geq v(x)-1$, the optimal interval is $[0;1-v(y)]$;
  \item if $v(y)\leq v(x)-1$, the optimal interval is $[0;2-v(x)]$.
  \end{itemize}
\end{itemize}
This defines the permissiveness for~$\loc_1$.

\medskip
We now look at~$\loc_0$: first, $\calP_1(\loc_0,v)=-\infty$ for
all~$v$, and only configurations~$(\loc_0,v)$ where $v(x)\leq 1$ and
$v(y)\leq 1$ are winning, so that $\calP_2(\loc_0,v)=-\infty$ as soon
as $v(x)>1$ or $v(y)>1$. Fix a valuation~$v$ for which $v(x)\leq 1$
and $v(y)\leq 1$. We~have to find the interval~$I=[\alpha,\beta]$ such
that $v(x)+\beta\leq 1$ and $v(y)+\beta\leq 1$, and for which
$\min\{\beta-\alpha, \inf_{\gamma\in[\alpha,\beta]}
\calP_1(\loc_1,(v+\gamma)[y\to 0])\}$ is maximized. Noticing that
$(v+\gamma)[y\to 0]$ is the valuation $(x\mapsto v(x)+\gamma; y\mapsto
0)$,
and that $\calP_1(\ell_1,w)=w(x)$ for any~$w$ satisfying $w(x)\in [0;1]$ and
$w(y)=0$,
we~have to maximize \( \min \{\beta-\alpha,
\inf_{\gamma\in[\alpha,\beta]} v(x)+\gamma\} \) over the domain
defined by $0\leq\alpha\leq\beta\leq \min(1-v(x);1-v(y))$.
Obviously, $\inf_{\gamma\in[\alpha,\beta]} v(x)+\gamma=v(x)+\alpha$,
so we have to maximize \( \min \{\beta-\alpha, v(x)+\alpha\} \)
on the set $\{(\alpha,\beta) \mid 0\leq\alpha\leq\beta\leq
\min(1-v(x);1-v(y))\}$.

\smallskip
We~consider two cases:
\begin{itemize}
\item if $v(y)\leq v(x)$: clearly, it is optimal to maximize~$\beta$, so we let
  $\beta=1-v(x)$. Hence we have to maximize $\min\{1-(v(x)+\alpha),
  v(x)+\alpha\}$ over $0\leq\alpha\leq 1-v(x)$. Again, there are two
  cases, depending on whether $v(x)$ is larger or smaller
  than~$1/2$; in~the former~case, $\min\{1-(v(x)+\alpha),
  v(x)+\alpha\} = 1-v(x)-\alpha$ when $\alpha$ ranges over $[0;
    1-v(x)]$; it~is maximized for~$\alpha=0$, and we get
  $\calP_2(\loc_0,v)=1-v(x)$.  If~$v(x)\leq 1/2$, the maximal value is
  reached when $\alpha=1/2-v(x)$, and $\calP_2(\loc_0,v)=1/2$.
\item if $v(y)\geq v(x)$: then it is optimal to let
  $\beta=1-v(y)$. Again there are two cases for maximizing
  $\min\{1-v(y)-\alpha, v(x)+\alpha\}$: if $1-v(y)\leq v(x)$, then
  $\alpha=0$ is optimal, and $\calP_2(\loc_0,v)=1-v(y)$; otherwise,
  $\alpha=(1-v(x)-v(y))/2$ is optimal, and
  $\calP_2(\loc_0,v)=(1-v(y)+v(x))/2$.
\end{itemize}
We~end up with the diagram represented on the left of
Fig.~\ref{fig-experm} (where for the sake of readability we~write $x$
and~$y$ in place of~$v(x)$ and~$v(y)$).
\end{example}

Our~aim in the rest of this paper is to compute the sequence of
functions~$\calP_i$, and to evaluate the complexity of this computation.
Following Lemma~\ref{lemma-converge}, this will provide us with an
algorithm for computing permissiveness in acyclic timed automata.




\section{Computing optimal strategies in linear timed automata}
\label{sec-linear}
%

In this section, we consider the simpler case of linear timed
automata, where each location has at most one successor.

\subsection{Optimal strategy for the opponent}

We~begin with focusing on the optimal choice of the opponent: given a
configuration~$(\loc,v)$ and an interval~$I$ of delays proposed by the
player
(there is a single outgoing transition, so the
action to be played is fixed), what is the best delay that the
opponent will choose so as to minimize the permissiveness of the
resulting configuration?

As we already mentioned, Lemma~\ref{lemma-decrdiag} answers this question
for non-resetting transitions: for such transitions,
the best option for the opponent
is to choose the maximal delay in the interval proposed by the player.
On~the other hand,
Example~\ref{ex-perm} provides a situation where the opponent prefers
to play as early as possible.

It~turns out that, for linear timed automata, the
optimal choice of the opponent is always one of these two extremal choices.
This property will be a corollary of the following lemma, stating
concavity of the  permissiveness function in linear timed automata:
\begin{restatable}{proposition}{propconcave}
  \label{prop-concave}
  Let~$i\in\bbN$.  Let $\loc$ be a location of a linear timed
  automaton, let $v_1$ and~$v_2$ be two clock valuations such that
  $\calP_i(\loc,v_1)$ and $\calP_i(\loc,v_2)$ are finite.
  Let~$\lambda\in[0;1]$, and $v_\lambda=\lambda\cdot v_1 +
  (1-\lambda)\cdot v_2$.  Then
\[
\calP_i(\loc,v_{\lambda}) \geq \lambda\cdot\calP_i(\loc,v_1) + (1-\lambda)\cdot\calP_i(\loc,v_2).
\]
\end{restatable}

The aim of the opponent being to select the valuation
in~$V=\{v+\delta[r\to 0] \mid 0\leq \delta\leq d\}$ that minimizes the
permissiveness. Writing~$v_1=v[r\to 0]$ and $v_2=v+d[r\to 0]$, we~have
$V=\{\lambda v_1 + (1-\lambda)v_2 \mid 0\leq \lambda\leq 1\}$.
Proposition~\ref{prop-concave} entails that the permissiveness is
minimized either in~$v_1$ or in~$v_2$. This~corresponds to our claim
that the best choice for the opponent always is to select one of the bounds
of the interval proposed by the player.

\begin{restatable}{corollary}{corostratPII}
  \label{coro-stratP2}
Let $\loc$ be a location of a linear timed automaton, $v$ and~$v'$ be two
clock valuations, $\lambda\in[0;1]$, and $v_\lambda=\lambda\cdot v +
(1-\lambda)\cdot v'$.
Then for all~$i$:
\[
\calP_i(\loc,v_{\lambda}) \geq \min\{\calP_i(\loc,v),\calP_i(\loc,v')\}.
\]
In particular, for any valuation~$v$, any bounded interval~$[\alpha,\beta]$, and any transition $\loc \trans[g,a,r] \loc'$:
\[
\inf \{\calP_i(\loc',v')\mid
\exists d\in [\alpha,\beta].\ (\loc,v)\trans[d,a](\loc',v')\}
= \min \{\calP_i(\loc',v'_\alpha), \calP_i(\loc',v'_\beta)\}
\]
where
$(\loc,v)\trans[\alpha,a](\loc',v'_\alpha)$ and
$(\loc,v)\trans[\beta,a](\loc',v'_\beta)$.
\end{restatable}

\subsection{Computing the most-permissive strategy}

Now that we have a better understanding of the optimal strategy of the
opponent, we can compute the most-permissive strategy of the player
for reaching the target location~$\loc_f$. We~prove that for all~$i$,
$\calP_i$~is in fact a piecewise-affine function that can be computed
in doubly-exponential time.

First notice that, following Lemma~\ref{lemma-converge}, for any
location~$\loc$ of a linear timed automaton with $n$ locations, the
sequence of functions $(\calP_i)_i$ converges in at most~$n$ steps.

\begin{theorem}\label{thm-main}
The permissiveness function for a linear timed automaton with $d$
locations and $n$ clocks is a piecewise-affine function.
It~can be computed in
time $O((n+1)^{8^d})$.
\end{theorem}

The following technical lemma
will be the central tool in the computation of~$\calP_i$:
\begin{restatable}{lemma}{lemmatech}
\label{lemma-technique}\label{lemma-tech}
Let $m_\alpha\leq M_\alpha$ and $m_\beta\leq M_\beta$, and
$D=\{(\alpha,\beta)\in\bbR+^2\mid m_\alpha\leq \alpha\leq M_\alpha,\
 m_\beta\leq \beta\leq M_\beta,\ \alpha\leq\beta\}$.
Let $f\colon \alpha\mapsto a\alpha+b$
and $g\colon \beta\mapsto c\beta+d$ be two 1-dimensional
affine functions, and
$\mu\colon (\alpha,\beta)\mapsto \min\{\beta-\alpha, f(\alpha), g(\beta)\}$.
Then the maximal value that $\mu$ may take over~$D$ is of one of the
following five forms: $M_\beta-m_\alpha$, $\lambda\cdot f(\nu)$,
$\lambda\cdot g(\mu)$, $\frac{ad-bc}{a-c}$ and
$\frac{ad-bc}{(a+1)(1-c)-1}$, with
$\lambda\in\{1,\frac1{1-c}, \frac1{a+1}\}$ and
$\nu\in\{m_\alpha,M_\alpha,m_\beta,M_\beta\}$.
This value can be computed by checking inequalities between expressions of
the same forms.
%
%
%
\end{restatable}

The following lemma corresponds to one step of our inductive
computation of~$\calP_i$:
%
\begin{restatable}{lemma}{lemmalinperm}
	\label{lemma-linperm}
  Let $\Aut$ be a linear timed automaton with~$n$ clocks.  Let
  $(\loc,g,a,z,\loc')$ be a transition of~$\Aut$, and assume that
  $v\mapsto \calP_{i-1}(\loc',v)$ is piecewise affine, with $m$ cells.
  Then $v\mapsto \calP_i(\loc,v)$ is piecewise affine. It~can be
  computed in time~$O(m^4\cdot (m+n)^4)$. It~can be defined using a
  polyhedral partition of size $O(m^4 \cdot(m+n)^4)$, and with coefficients
  polynomial in those of~$\calP_{i-1}$.
%
\end{restatable}

\begin{proof}
%
%
%
%
  We~assume that $\calP_{i-1}(\loc',v)$ is not
constantly~$-\infty$ (if~it were the case, then also
$\calP_i(\loc,v)=-\infty$ for all~$v$). Similarly, we~assume that
$\moves(\loc,v)$ is non-empty for some~$v$.
Since $v\mapsto \calP_{i-1}(\loc',v)$ is piecewise-affine:
we~can then fix a
polyhedral partition $\polyh{\Phi,P}$ and, for each cell~$h$ in this
partition, an affine functions $f_h$, such that
$\calP_{i-1}(\loc',v)=f_h(v)$ for the only cell~$h$ containing~$v$.

\begin{figure}[b]
  \centering

  \begin{tikzpicture}[scale=.7]
    \begin{scope}
      \draw[latex'-latex'] (5,0) -| (0,5);
      \draw[vert] (3,3) -- (4,2) -- (4.75,4) -- (4,4) -- cycle;
      \path (3.9,2.5) node {$h_\beta$};
      \draw[bleu] (1,1) -- (1,1.5) -- (2,3) -- (2.5,1.5) -- cycle;
      \path (1.8,2.2) node {$h_\alpha$};
      \draw[dashed,opacity=.5] (5,4) -- (1,0);
      \draw[dashed,opacity=.5] (5,5) -- (0,0);
      \draw[rouge,dashed,line width=1pt,opacity=.5] (0,0) -- (2.25,2.25) --
      (2.5,1.5) -- (1,0) -- cycle;
      \draw[rouge,dashed,line width=1pt,fill=none] (0,0) -- (2.25,2.25) --
      (2.5,1.5) -- (1,0) -- cycle;
      \draw (1,.5) edge[out=-20,in=-160,latex'-]  node[pos=1,right]
            {$S_{(h_\alpha,h_\beta)}$} (3,.5);
    \end{scope}
    \begin{scope}[xshift=6cm]
      \draw[latex'-latex'] (5,0) -| (0,5);
      \draw[vert] (3,3) -- (4,2) -- (4.75,4) -- (4,4) -- cycle;
      \path (3.9,2.5) node {$h_\beta$};
      \draw[bleu] (1,1) -- (1,1.5) -- (2,3) -- (2.5,1.5) -- cycle;
      \path (1.8,2.2) node {$h_\alpha$};
      \draw[dashed,opacity=.5] (5,5) -- (0,0);
      \draw[rouge,dashed,line width=1pt,opacity=.15,fill=none] (0,0) -- (2.25,2.25) --
      (2.5,1.5) -- (1,0) -- cycle;
      \draw[dashed,opacity=.5] (5,4.25) -- (.75,0);
      \draw[dashed,opacity=.15] (5,4) -- (1,0);
      \draw[rouge,dashed,line width=1pt,opacity=.5] (0,0) -- (2.25,2.25) --
      (2.25,1.5) -- (.75,0) -- cycle;
      \draw[rouge,dashed,line width=1pt,fill=none] (0,0) -- (2.25,2.25) --
      (2.25,1.5) -- (.75,0) -- cycle;
      \draw (1,.5) node[inner sep=0pt,fill=black,minimum size=3pt,circle]
      (v) {} node[below left] {$v$};
      \draw[dashed,opacity=.5] (v) -- +(4,4);
      \draw[line width=1pt,bleu] (1.75,1.25) node[inner sep=0pt,bleu,minimum size=3pt,circle] {} -- (2.375,1.875) node[inner sep=0pt,bleu,minimum size=3pt,circle] {} node[midway,coordinate] (b) {};
      \draw[line width=1pt,vert] (3.25,2.75) node[inner sep=0pt,vert,minimum size=3pt,circle] {} -- (4.5,4) node[inner sep=0pt,vert,minimum size=3pt,circle] {}
      node[midway,coordinate] (g) {};
      \draw (b) edge[latex'-,out=-60,in=180]
        node[right,pos=1] {$I^v_{\alpha}$} (3,1);
      \draw (g) edge[latex'-,out=120,in=0]
        node[left,pos=1] {$I^v_{\beta}$} (2.5,4);
    \end{scope}
    \begin{scope}[xshift=12cm]
      \draw[latex'-latex'] (5,0) -| (0,5);
      \draw[vert] (3,3) -- (4,2) -- (4.75,4) -- (4,4) -- cycle;
      \path (3.9,2.5) node {$h_\beta$};
      \draw[bleu] (1,1) -- (1,1.5) -- (2,3) -- (2.5,1.5) -- cycle;
      \path (1.8,2.2) node {$h_\alpha$};
      \draw[dashed,opacity=.5] (5,4.25) -- (.75,0);
      \draw[dashed,opacity=.15] (5,4) -- (1,0);
      \draw[dashed,opacity=.5] (5,5) -- (0,0);
      \draw[rouge,dashed,line width=1pt,opacity=.5] (0,0) -- (2.25,2.25) --
      (2.25,1.5) -- (.75,0) -- cycle;
      \draw[rouge,dashed,line width=1pt,fill=none] (0,0) -- (2.25,2.25) --
      (2.25,1.5) -- (.75,0) -- cycle;

      \draw (1,.5) node[inner sep=0pt,fill=black,minimum size=3pt,circle]
      (v) {} node[below left] {$v$};
      \draw[dashed,opacity=.5] (v) -- +(4,4);
      \draw[line width=1pt,bleu]
      (1.75,1.25) node[inner sep=0pt,bleu,minimum size=3pt,circle] {};
      \draw[line width=1pt,vert] (3.25,2.75) node[inner sep=0pt,vert,minimum size=3pt,circle] {};
      \draw[opacity=.5,black,line cap=round,line width=3pt] (1.75,1.25) -- (3.25,2.75) node[pos=.7,coordinate] (m) {};
      \draw (m) edge[latex'-,out=-30,in=90]  node[pos=1,below,text width=2cm,scale=.6,align=center]
        {interval to be played} (3.5,1.2);
    \end{scope} 
  \end{tikzpicture}
        
  \caption{Three steps of our procedure:
    $S_{(h_\alpha,h_\beta)}$; then compute expressions for
    $I^v_\alpha$ and $I^v_\beta$ (notice that we had to refine
    $S_{(h_\alpha,h_\beta)}$, because the expression for $I^v_\beta$
    would be different for the lower part of $S_{(h_\alpha,h_\beta)}$ since
    it ends of a different facet of~$h_\beta$); 
    finally select best values for~$\alpha$ and~$\beta$.}
  \label{fig-procedure}
\end{figure}
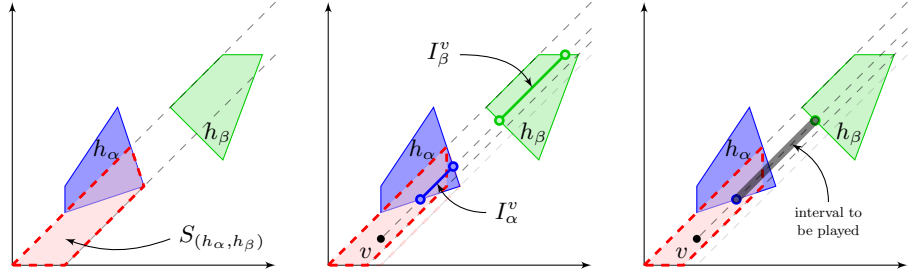

Our procedure for computing~$\calP_i$ in~$\loc$ consists in listing
the possible pairs of cells defining~$\calP_{i-1}$ in~$\loc'$ where
the left- and right-bounds of the interval to be proposed lie.
For each pair~$(h_\alpha,h_\beta)$ of such cells,
we~perform the following three steps (illustrated on
Figure~\ref{fig-procedure}):
\begin{itemize}
\item characterize the set~$S_{(h_\alpha,h_\beta)}$ of all
  valuations from which those cells can be
  reached by taking the transition from~$\loc$ to~$\loc'$. We~compute
  this polyhedron using quantifier elimination;
\item compute the ranges for~$\alpha$ and~$\beta$ that
  can be played in order to indeed end up respectively in~$h_\alpha$
  and~$h_\beta$. These are intervals~$I_\alpha$ and $I_\beta$, whose
  bounds are expressed as functions of~$v$.
  Computing these bounds may require refining the polyhedron obtained at
  the previous step into several subpolyhedra, in order to
  express them as affine functions of~$v\in
  S_{(h_\alpha,h_\beta)}$;
\item for each subpolyhedron, compute the optimal values for~$\alpha$
  and~$\beta$: following Corollary~\ref{coro-stratP2}, this amounts to
  find values for~$\alpha\in I_\alpha$ and~$\beta\in I_\beta$ that
  maximize the following function:
  \[
  \mu\colon (\alpha,\beta) \mapsto \min\{\beta-\alpha;
  \calP_{i-1}(\loc',(v+\alpha)[z\to 0]);
  \calP_{i-1}(\loc',(v+\beta)[z\to 0])  \} . 
  \]
  This is performed by applying our technical Lemma~\ref{lemma-tech};
  it may again require another refinement of the subpolyhedra, and
  returns an affine function for each subpolyhedron.
\end{itemize}
For each pair~$(h_\alpha,h_\beta)$, we~end up with a (partial)
piecewise-affine function, defined on $S_{(h_\alpha,h_\beta)}$,
returning the optimal permissiveness that can be obtained if playing
interval~$[\alpha,\beta]$ such that taking the transition to~$\loc'$
after delay~$\alpha$ (resp.~$\beta$) leads to~$h_\alpha$
(resp.~$h_\beta$).
Our final step to compute~$\calP_i$ in~$\loc$ consists in
taking the maximum of all these partial functions on their (possibly
overlapping) domains; this may introduce on more refinement of our
polyhedron.

Notice that all these computations are performed symbolically
w.r.t~$v$: we~manipulate affine functions of~$v$, with conditions
on~$v$ for our computations to be valid.
\end{proof}

Assuming that $v\mapsto \calP_{i-1}(\loc',v)$ has $m$ cells,
computing~$v\mapsto \calP_i(\loc,v)$ takes time~$O(m^4\cdot (m+n)^4)$,
where $n$ is the number of clocks, and this function has $O(m^4\cdot
(m+n)^4)$ many cells. 

It~follows that, for a linear timed automaton having~$d$ locations,
we~obtain the permissiveness function in the initial state as a
piecewise-affine function in time~$O((n+1)^{8^d})$, which proves
Theorem~\ref{thm-main}.

This complexity is quite high, but it~is a rough approximation.

In~Appendix~\ref{app-example-long}, we~develop a complete computation
of~$\calP_i$ on the linear timed automaton of Fig.~\ref{fig-experm2}
(which only differs from the example of Fig.~\ref{fig-experm} in the guard of
the first transition); in~this computation, we have many intermediary
cases to handle, but the final function~$\calP_2$ in~$\loc_0$,
depicted on Fig.~\ref{fig-permexlong}, has a partition with only four
cells (in~the winning~zone).
\begin{figure}[!ht]
  \centering
  \begin{tikzpicture}
    \begin{scope}
      \draw (0,0) node[rond,bleu] (a) {$\loc_0$};
      \draw (4,0) node[rond,rouge] (b) {$\loc_1$};
      \draw (8,0) node[rond,vert] (c) {$\loc_f$};
      \draw (a) edge[draw,-latex'] node[above] {$\begin{array}{>{\scriptstyle}c}
          0\leq y\leq 1
        \end{array}$} node[below] {$\scriptstyle y:=0$} (b);
      \draw (b) edge[draw,-latex'] node[above] {$\begin{array}{>{\scriptstyle}c}
          1\leq x\leq 2 \\ 0\leq y\leq 1
        \end{array}$} (c);
    \end{scope}
  \end{tikzpicture}
  \caption{Automaton of Fig.~\ref{fig-experm} where the guard on the
    first transition has been slightly extended}
  \label{fig-experm2}

  \begin{tikzpicture}
    \begin{scope}[xshift=-1cm,yshift=-3cm,scale=3]
      \fill[fill=blue!40!white] (.33,0) -| (.66,.33) -- cycle;
      \shade[left color=blue!40!white, right color=blue!15!white]
        (.66,0) -- (.66,.33) -- (2,1) |- cycle;
      \shade[shading=axis,bottom color=blue!45!white,top color=blue!20!white,shading angle=45]
      (0,0) -- (.33,0) -- (.66,.33) -- (0,1) -- cycle;
      \shade[bottom color=blue!40!white,top color=blue!15!white]
      (0,1) -- (.66,.33) -- (2,1) -- cycle;

      \path[use as bounding box] (0,1.2) -- (0,-.1);

      \begin{scope}
      \path[clip] (0,0) rectangle (2.2,1.2);
      \foreach \x in {
          1,
          1.33,
          1.66,
          2} 
        {\draw[blue!70!green!40!white] (\x,0) -- (\x,\x/2) -- (1-\x/2,\x/2) -- (0,\x-1);}
      \end{scope}

            \draw[latex'-latex'] (2.2,0) node[right] {$x$} -| (0,1.2)
        node[left] {$y$};
      \begin{scope}
        \path[clip] (-.5,-.5) rectangle (2.2,1.2);
      \foreach \x in {0,1,2}
           {\draw (\x,0) -- +(0,-.05) node[below] {$\scriptstyle \x$};
            \draw (0,\x) -- +(-.05,0) node[left] {$\scriptstyle \x$};
            \draw[dashed] (\x,0) -- +(0,1.2);
            \draw[dashed] (0,\x) -- +(2.2,0);
           }
      \draw (.33,0) -- (.66,.33) -- (.66,0);
      \draw (0,1) -- (.66,.33) -- (2,1);
      \end{scope}

      \draw (.58,.1) node {$\frac23$};
      \path (1.35,.3) node {$\scriptstyle 1-\frac x2$};
      \draw (1,.8) node {$\scriptstyle 1-y$};
      \draw (.22,.4) node {$\frac{1+x-y}2$};
      \draw (.9,1.1) node {$\scriptstyle -\infty$};
      \draw (2.1,.5) node {$\scriptstyle -\infty$};
    \end{scope}
  \end{tikzpicture}
  \caption{Permissiveness function of the automaton of Fig.~\ref{fig-experm2}
  in~$\loc_0$}
  \label{fig-permexlong}
\end{figure}

%
%
%
%






\section{Extension to acyclic timed automata and games}
\label{sec-generic}
\subsection{Adding branching}

We extend the previous study to the case of acyclic timed automata
(with branching).
%
In that case, we can still apply our inductive approach, with a few
changes: at each step, we would compute the optimal move of the player
for each single action, and then select the optimal action by
``superimposing'' the resulting permissiveness functions and selecting
the action that maximizes permissiveness. This however
breaks the result of Prop.~\ref{prop-concave}:
the maximum of two concave functions
need not be concave. Example~\ref{ex-nonconcave}, derived from
Example~\ref{ex-perm}, displays an example where the permissiveness
function is not concave.

\begin{example}\label{ex-nonconcave}
  Consider the automaton of Fig~\ref{fig-exnonconcave}.  The
  transition from~$\loc_0$ to~$\loc_f$ has the same constraint as that
  from~$\loc_1$ to~$\ell_f$; hence the permissiveness offered by that
  action is the same as the one from~$\loc_1$, which we already
  computed. Hence the global permissiveness from~$\loc_0$ is the
  (pointwise) maximal of the two piecewise-affine functions displayed
  on Fig.~\ref{fig-experm}, which is depicted on
  Fig.~\ref{fig-exnonconcave}. On this diagram, the blue area
  corresponds to points from where it is better (or only possible) to
  go via~$\loc_1$, while the red area corresponds to valuations from
  where it is better (or only possible) to take the bottom transition.

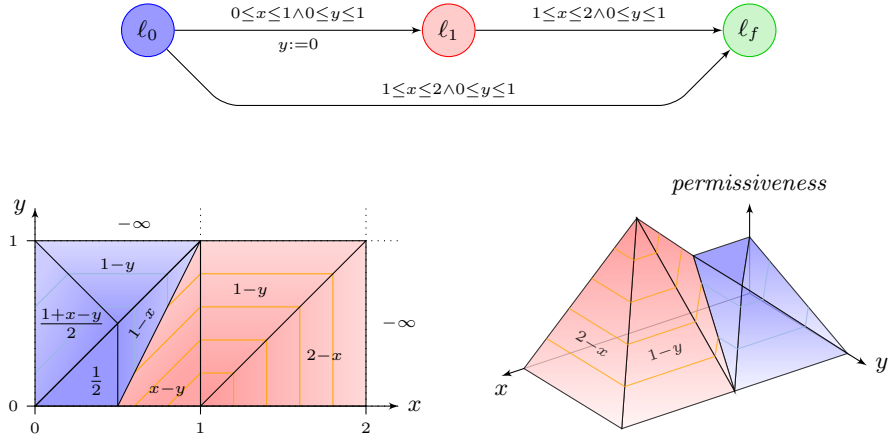
\begin{figure}[!ht]
  \centering
  \begin{tikzpicture}
    \begin{scope}
      \draw (0,0) node[rond,bleu] (a) {$\loc_0$};
      \draw (4,0) node[rond,rouge] (b) {$\loc_1$};
      \draw (8,0) node[rond,vert] (c) {$\loc_f$};
      \everymath{\scriptstyle}
      \draw (a) edge[draw,-latex'] node[above] {$0\leq x\leq 1 \wedge 0\leq y\leq 1$} node[below] {$\scriptstyle y:=0$} (b);
      \draw (b) edge[draw,-latex'] node[above] {$1\leq x\leq 2 \wedge 0\leq y\leq 1$} (c);
      \draw[rounded corners=2mm,-latex'] (a) -- +(1,-1) -- node[above] {$1\leq x\leq 2 \wedge 0\leq y\leq 1$} ($(c)+(-1,-1)$) -- (c);
    \end{scope}
    \begin{scope}[xshift=-1.5cm,yshift=-5cm,scale=2.2]
      \shade[shading=axis, top color=red!0!white, bottom color=red!35!white, shading angle=45]
      (0,0) -| (1,1) -- cycle;
      \shade[shading=axis, bottom color=red!40!white, top color=red!15!white]
      (1,0) |- (2,1) -- cycle;
      \shade[shading=axis, left color=red!40!white, right color=red!15!white]
      (1,0) -| (2,1) -- cycle;

      \foreach \x in {.2,.4,.6,.8} 
        {\draw[red!35!yellow] (\x,0) -- (1,1-\x) -- (2-\x,1-\x) -- (2-\x,0);}

      \fill[fill=blue!40!white] (0,0) -| (.5,.5) -- cycle;
      \shade[left color=blue!40!white, right color=blue!15!white]
        (.5,0) -- (.5,.5) -- (1,1) -- (.5,0) -- cycle;
      \shade[shading=axis,bottom color=blue!45!white,top color=blue!20!white,shading angle=45]
      (0,0) -- (.5,.5) -- (0,1) -- cycle;
      \shade[bottom color=blue!40!white,top color=blue!15!white]
      (0,1) -- (.5,.5) -- (1,1) -- cycle;

      \foreach \x in {.6,.8} 
        {\draw[blue!70!green!40!white] (\x,2*\x-1) -- (\x,\x) -- (1-\x,\x) -- (0,2*\x-1);}

      \draw[latex'-latex'] (2.2,0) node[right] {$x$} -| (0,1.2) node[left] {$y$};
      \foreach \x in {0,1}
           {\draw (\x,0) -- +(0,-.05) node[below] {$\scriptstyle \x$};
            \draw (0,\x) -- +(-.05,0) node[left] {$\scriptstyle \x$};
            \draw[dotted] (\x,0) -- +(0,1.2);
            \draw[dotted] (0,\x) -- +(2.2,0);
           }
           \draw (2,0) -- +(0,-.05) node[below] {$\scriptstyle 2$};
           \draw[dotted] (2,0) -- +(0,1.2);

      \draw (.37,.15) node {$\frac12$};
      \path (.3,0)--(1,1) node[midway,sloped] {$\scriptstyle 1-x$};
      \draw (.5,.85) node {$\scriptstyle 1-y$};
      \draw (.22,.5) node {$\frac{1+x-y}2$};
      \path (.8,.1) node {$\scriptstyle x-y$};
      \path (1.75,.3) node {$\scriptstyle 2-x$};
      \path (1.3,.7) node {$\scriptstyle 1-y$};
      \draw (.6,1.1) node {$\scriptstyle -\infty$};
      \draw (2.2,.5) node {$\scriptstyle -\infty$};
      \draw (0,0) -- (1,1) -- (2,1) -- (1,0) -- cycle; 
      \draw (1,1) -- (1,0) -| (2,1);
      \draw (0,0) -- (1,1);
      \draw (.5,0) -- (.5,.5) -- (0,1);
      \draw (1,0) |- (0,1);
      \draw (.5,0) -- (1,1);
    \end{scope}

    \begin{scope}[xshift=8cm,yshift=-3.5cm,scale=1]
    \begin{scope}[x={(-1.5cm,-.5cm)},y={(1.3cm,-.8cm)},z={(0cm,1.5cm)}]
    \draw[-latex'] (0,0,0) -- +(2.2,0,0) node[below] {$x$};
    \draw[-latex'] (0,0,0) -- +(0,1.2,0) node[right] {$y$};
    \draw[-latex'] (0,0,0) -- +(0,0,.8) node[above] {\textit{permissiveness}};
    \begin{scope}[opacity=.8]
    \draw[fill=blue!45!white] (0,0,.5) -- (.5,.5,.5) -- (.5,0,.5) -- cycle;
    \shade[draw,bottom color=blue!5!white,top color=blue!45!white]
      (0,0,.5) -- (.5,.5,.5) -- (0,1,0) -- cycle;
    \shade[draw,bottom color=blue!5!white,top color=blue!45!white]
      (0,1,0) -- (.5,.5,.5) -- (1,1,0) -- cycle;
    \shade[draw,bottom color=blue!5!white,top color=blue!45!white]
    (1,1,0) -- (.5,0,.5) -- (.5,.5,.5) -- cycle;
    \shade[draw,bottom color=red!5!white,top color=red!50!white]
    (1,1,0) -- (.5,0,.5) -- (1,0,1) -- cycle;
    \shade[draw,bottom color=red!5!white,top color=red!50!white]
    (1,1,0) -- (1,0,1) -- (2,1,0) -- cycle;
    \shade[draw,bottom color=red!5!white,top color=red!50!white]
    (1,0,1) -- (2,1,0) -- (2,0,0) -- cycle;
    \everymath{\scriptstyle}
    \path (1.7,0,.3) -- (1.7,.7,.3) node[midway,sloped] {$2-x$};
    \path (1,.7,.3) -- (1.7,.7,.3) node[midway,sloped] {$1-y$};
    \foreach \x in {.6,.8} 
      {\draw[blue!70!green!40!white] (\x,2*\x-1,1-\x) -- (\x,\x,1-\x) -- (1-\x,\x,1-\x) -- (0,2*\x-1,1-\x);}
    \foreach \x in {.2,.4} 
      {\draw[red!35!yellow] (1-\x,1-2*\x,\x) -- (1,1-\x,\x) -- (2-\x,1-\x,\x) -- (2-\x,0,\x);}
    \foreach \x in {.6,.8} 
      {\draw[red!35!yellow] (\x,0,\x) -- (1,1-\x,\x) -- (2-\x,1-\x,\x) -- (2-\x,0,\x);}
    \end{scope}
    \end{scope}
    \end{scope}
  \end{tikzpicture}
  \caption{A timed automaton and its (non-concave) permissiveness
    function in~$\loc_0$}
  \label{fig-exnonconcave}
\end{figure}
\end{example}


We prove by induction that the permissiveness functions still are
piecewise-affine in that setting.  Hence all four steps of our proof
of Lemma~\ref{lemma-linperm} still apply, with some adaptations.
For~each location~$\loc$, for~each transition~$t$ from~$\loc$ to some~$\loc'$,
the~procedure now is as follows:
\begin{itemize}
\item for the first step, we~again consider two cells~$h_\alpha$
  and~$h_\beta$ in the partition defining~$\calP_{i-1}(\loc')$,
  together with a set~$H$ of cells that will be visited
  between~$h_\alpha$ and~$h_\beta$. Again applying Fourier-Motzkin,
  we~get a polyhedron~$S_{(h_\alpha,h_\beta,H)}$ of valuations from
  which those cells can indeed be visited;
\item the computation of the intervals~$I^v_\alpha$ and $I^v_\beta$ is
  unchanged;
\item for each cell~$h\in H$, we~can compute the values~$d_h^{\sfin}$
  and $d_h^{\sfout}$ for which $(v+d_h^{\sfin})[z\to 0]$ enters~$h$ and
  $(v+d_h^{\sfout})[z\to 0]$ leaves~$h$ (notice that this may require
  further refinement of the polyhedron being considered). Since
  $\calP_{i-1}$ is affine on cell~$h$, it~reaches its maximum on this
  cell either at $(v+d_h^{\sfin})[z\to 0]$ or at $(v+d_h^{\sfout})[z\to
    0]$. The~function we need to maximize now looks like
  \begin{multline*}
  \mu'\colon (\alpha,\beta) \mapsto \min(\{\beta-\alpha,
  \calP_{i-1}(\loc',(v+\alpha)[z\to 0]),
  \calP_{i-1}(\loc',(v+\beta)[z\to 0])\} \cup \\
  \{\calP_{i-1}(\loc',(v+d_h^{\sfin})[z\to 0]),
  \calP_{i-1}(\loc',(v+d_h^{\sfout})[z\to 0]) \mid h\in H\}).
  \end{multline*}
  Now, we notice that all values in the second set are constant, not
  depending on~$\alpha$ and~$\beta$. We~can thus still apply
  Lemma~\ref{lemma-tech} in order to maximize $\mu(\alpha,\beta)$, and
  then take the above constants into account (which may again refine
  the polyhedra). 
\item the above three steps have to be performed for all outgoing
  transitions from the location~$\loc$ being considered. The~last step
  still consists in selecting the maximum of all the resulting functions.
\end{itemize}

The complexity of our procedure is much higher than that of linear
automata: because we consider sets of cells already at the first step,
we~have $O(2^m)$ sets to consider. Assuming that~$\calP_{i-1}$ is made
of $m$~cells, we~may end up with $\calP_i$ having more than $2^m$ cells.
Since we~have to repeat this procedure  up to~$|T|$ times, so that
the time complexity is in~$O(^{|T|}2)$ (where $^na$ is tetration).
Hence our procedure is
non-elementary in the worst case.
In~the end:
\begin{theorem}
  The permissiveness function for acyclic timed automata
  is piecewise affine. It~can be computed in non-elementary time.
\end{theorem}

\subsection{Adding uncontrollable states}
We~finally extend our approach to (acyclic) two-player turn-based
timed games.

This~setting is easily seen to preserve piecewise-affineness of the
permissiveness function.  Indeed, in~order to compute~$\calP_i$ in a
location~$\loc$ belonging to the opponent, it~suffices to first
compute the functions~$\calP_{i}^{\loc\to\loc'}$ for all outgoing
transitions from~$\loc$ to some~$\loc'$; this follows the same
procedure as above, and results in a piecewise-affine function,
assuming (inductively) that $\calP_{i-1}$~is piecewise affine.
We~then compute the (still piecewise-affine)
minimum~$\calM_{i}(\loc,v)$ of all those functions, and finally
\[
\calP_i(\loc,v) = \min_{\substack{d \text{ s.t. }\\ v+d\models \sfInv(\loc)}}
\calM_{i}(\loc,v+d)
\]
which is easily computed and remains piecewise-affine.
The~computation for locations that belong to the player is similar as
in the case of plain timed automata.
It follows:
\begin{theorem}
  The permissiveness function for acyclic turn-based timed games 
  is piecewise affine, and can be computed in non-elementary time.
\end{theorem}


\section{Conclusions and perspectives}
\label{sec-conclu}
\looseness=-1
In this paper, we addressed the problem of measuring the amount of
precision needed in a timed automaton to reach a given target
location. We~built on the formalism of permissive strategies defined
in~\cite{BFM15}, and developed an algorithm for computing the optimal
permissiveness in acyclic timed automata and games. 

\looseness=-1
There are several directions in which we will extend this work: as a
first task, we will have a closer look at the complexity of our
procedure, trying to either find examples where the number of cells
indeed grows exponentially (for linear timed automata) or
exponentially at each step (for acyclic timed automata). A~natural
continuation of our work consists in tackling cycles. We~were unable
to prove our intuition that there is no reason for the player to
iterate a cycle.  Following~\cite{BMRS19}, we~might first consider
fixing a timed automaton made of a single cycle, study how
permissiveness evolves along one run in this cycle, and compute the
optimal permissiveness for being able to take a cycle forever.
Exploiting 2-Lipschitz continuity of the permissiveness function,
we~could also develop approximating techniques, both for making our
computations more efficient in the acyclic case and to handle cycles.
%
Finally, other interesting directions include extending our approach
to linear hybrid automata, or considering a
stochastic opponent, thereby modelling the fact that perturbations
need not always be antagonist.

\bibliographystyle{alpha}
\bibliography{biblio}

\clearpage
\appendix

\section{Proofs of Section~\ref{sec-defs}}
\label{lemma-proofs}
This section is devoted to the proofs of the lemma of Section~\ref{sec-defs}.

\lemmaincri*
\begin{proof}
	The proof is by induction. For any configuration~$(\loc,v)$, either
	$\calP_0(\loc,v)=\calP_1(\loc,v)=+\infty$, or
	$\calP_0(\loc,v)=-\infty$. In~both cases, $\calP_0(\loc,v)\leq
	\calP_1(\loc,v)$.
	
	Then, assuming $\calP_i(\loc,v)\leq \calP_{i+1}(\loc,v)$ for
	all~$(\loc,v)$, we~directly get the same property at step~$i+1$. The
	result follows.
\end{proof}


\lemmaconverge*
\begin{proof}
	By~induction on~$i$: for~$i=0$, only~$\loc_f$ satisfies the condition,
	and the result holds by definition of~$\calP_i$ for~$\loc_f$.
	
	Now, assume that the result holds for some index~$i$, and consider a
	location~$\loc$ such that the longest path to~$\loc_f$ has at
	most~$i+1$ transitions. Then any successor location~$\loc'$ of~$\loc$
	has longest path of length at most~$i$, hence
	$\calP_i(\loc',v)=\calP_{i+1}(\loc',v)$. It~immediately follows that
	$\calP_{i+1}(\loc,v)=\calP_{i}(\loc,v)$ for any~$v$. 
\end{proof}

\lemmainfini*

\begin{proof}
	The hypotheses ensure that any action- and delay transition
	performed from~$(\loc,v)$ can be performed from~$(\loc,v')$, and the
	resulting configurations still satisfy the conditions of the lemma.
	The~result follows by~induction.
\end{proof}

\lemmadecrdiag*
\begin{proof}
	For any move~$(I,a)$ that is available from~$(\ell,v+t)$, the~move
	$(I+t,a)$ is available from~$(\ell,v)$. Moreover, the set of
	valuations on which $\calP_{i-1}$ is minimized is the same in both cases,
	namely $\{(v+t[z\to 0]) \mid d\in I\}$. 
	It~follows that $\calP_i(\ell,v+t) \leq \calP_i(\ell,v)$.
	
	Conversely, for any move~$(I,a)$ available from~$(\ell,v)$ with 
	$|I|\geq t$ (if~any), the move $((I-t)\cap \bbR+,a)$ is a valid
	(non-empty) move from~$(\ell,v+t)$. The second inequality follows.
\end{proof}

Finally, to prove proposition \ref{prop-continuite}, we use the following lemmas:


   \begin{lemma} \label{lemma-I}
     Let~$v$ and~$v'$ be two clock valuations. Write
     $\eta=\norm{v'-v}$. If $([\alpha,\beta],a)\in\moves(\loc,v)$ with
     $\beta-\alpha\geq 2\eta$, then $([\alpha+\eta,\beta-\eta],a) \in
     \moves(\loc,v')$.
   \end{lemma}
   \begin{proof}
     For any clock~$c$, $0\leq v(c)+\alpha = v'(c)+\alpha
     +(v(c)-v'(c)) \leq v'(c)+\alpha+\eta$. Similarly, $v(c)+\beta
     =v'(c)+\beta +(v(c)-v'(c)) \geq v'(c)+\beta-\eta$. Then for any
     interval~$J$, if $v(c)+[\alpha,\beta] \subseteq J$, and also
     $v'(c)+[\alpha+\eta,\beta-\eta]\subseteq J$. It~follows that for
     any guard~$g$, if $v+[\alpha,\beta]\subseteq g$ with
     $\beta-\alpha \geq 2\norm{v'-v}$, then $v'+ [\alpha+\norm{v'-v},
       \beta-\norm{v'-v}] \subseteq g$. 
   \end{proof}

   \begin{corollary}\label{coro-I}
     For any integer~$i\in\bbN$ and any location~$\loc$, the function
     $\nu_\loc\colon v\mapsto \sup_{(I,a)\in\moves(\loc,v)} |I|$
     is $2$-Lipschitz continuous
     on the set $\{v \mid \moves(\loc,v)\not=\varnothing\}$.
   \end{corollary}
   \begin{proof}
     We~first prove the result for the case where location~$\loc$ has a single
     transition~$(\loc,g,a,z,\loc')$.  Take two valuations~$v$
     and~$v'$ for which $\moves(\loc,v)$ is non-empty.  We~prove that
     $\nu_\loc(v')-\nu_\loc(v)\geq -2\norm{v'-v}$. By~symmetry of the
     roles of~$v$ and~$v'$, our result (for a single outgoing
     transition) follows.

     First, if $\moves(\loc,v)$ contains no
     intervals of size at least $2\norm{v'-v}$, then obviously
     $\nu_\loc(v') - \nu_\loc(v) \geq -2 \norm{v'-v}$.
     
     Now, assume that there exists $[\alpha,\beta]\in\moves(\loc,v)$
     such that $\beta-\alpha\geq 2\norm{v'-v}$.
     By~Lemma~\ref{lemma-I}, for any such interval, the interval
     $[\alpha+\norm{v'-v}, \beta-\norm{v'-v}]\in\moves(\loc,v')$.

     Fix~$\varepsilon>0$, and take $I=[\alpha,\beta]\in\moves(\loc,v)$
     such that $|I| \geq \nu_\loc(v)-\varepsilon$. Since
     $[\alpha+\norm{v'-v}, \beta-\norm{v'-v}]\in\moves(\loc,v')$,
     it~follows $\nu_\loc(v') \geq \nu_\loc(v) - 2\norm{v'-v}
     -\varepsilon$. Since this holds for any~$\varepsilon>0$, we~get
     the announced inequality.

     Now, in case there are several outgoing transitions, we~have
     \[
     \nu_\loc(v) = \sup_{(I,a)\in\moves(\loc,v)} |I| =
     \max_{a\in\Sigma} \sup_{(I,a)\in\moves(\loc,v)} |I|.
     \]
     Hence
     $\nu_\loc$ is the pointwise maximum of $2$-Lipschitz continuous
     functions, hence it is $2$-Lipschitz continuous.
   \end{proof}

\propcontinuite*

   \begin{proof}
     The proof is again by induction on~$i$. The case of~$i=0$ is
     trivial. Corollary~\ref{coro-I} proves the result for~$i=1$.

     Now, assume that the result holds for some index~$i-1$.  Take a
     location~$\loc$, and again first assume that $\loc$~has a single
     outgoing transition~$(\loc,g,a,z,\loc')$. As~in the previous
     proof, the result for the general case directly follows.
     
     Pick two valuations~$v$ and~$v'$ such that
     $\calP_i(\loc,v)$ and $\calP_i(\loc,v')$ are
     finite. In~particular, $\moves(\loc,v)$ and $\moves(\loc,v')$ are
     non-empty. We~follow the same approach as in the proof of
     Corollary~\ref{coro-I}, proving that $\tau_\loc(v')-\tau_\loc(v)\geq
     -2\norm{v'-v}$. By~symmetry, our result follows.

     Again, in case $\moves(\loc,v)$ contains no intervals of size
     larger than or equal to~$\norm{v'-v}$, the result is immediate.
     Otherwise, fix~$\varepsilon>0$, and take an
     interval~$I=[\alpha,\beta]$ such that
     \[
     \min(|I|, \inf_{d\in I} (\calP_{i-1}(\loc',(v+d)[z\to 0]))) \geq
     \tau_\loc(v) - \varepsilon.
     \]
     Then $|I|\geq \tau(v)-\varepsilon$ and for any~$d\in I$,
     $\calP_{i-1}(\loc',(v+d)[z\to 0]) \geq \tau_\loc(v)-\varepsilon$.

     Let~$I'=[\alpha+\norm{v'-v}, \beta-\norm{v'-v}]$. Then
     \[|I'|\geq
     |I|-2\norm{v'-v}\geq \tau_\loc(v)-\varepsilon-2\norm{v'-v}.
     \]
     Moreover, since $I'\subseteq I$, we~have
     $\calP_{i-1}(\loc',(v+d)[z\to 0]) \geq \tau_\loc(v)-\varepsilon$
     also when~${d\in I'}$. 
     Additionally, for any~$d\in I'$,
     \[
     \norm{(v'+d)[z\to 0] - (v+d)[z\to 0]} \leq \norm{v'-v},
     \]
     so that
     \begin{xalignat*}1
       \multicolumn{2}{l}{$\calP_{i-1}(\loc',(v+d)[z\to 0]) - \calP_{i-1}(\loc',(v'+d)[z\to 0])$} \\
       \hskip.45\linewidth & {}\leq 2\norm{(v'+d)[z\to 0] - (v+d)[z\to 0]} \\
     &{} \leq     2\norm{v'-v}.
     \end{xalignat*}
     Thus for any~$d\in I'$,
     \begin{xalignat*}1
     \calP_{i-1}(\loc',(v'+d)[z\to 0]) &\geq
     \calP_{i-1}(\loc',(v+d)[z\to 0]) -2\norm{v'-v} \\&
     \geq
     \tau_\loc(v)-\varepsilon-2\norm{v'-v}.
     \end{xalignat*}
     Since also $|I'|\geq
     \tau_\loc(v)-\varepsilon-2\norm{v'-v}$, we get
     \[
     \tau_\loc(v') \geq
     \min(|I'|, \inf_{d\in I'} (\calP_{i-1}(\loc',(v'+d)[z\to 0]))) \geq
     \tau_\loc(v)-\varepsilon-2\norm{v'-v}.
     \]
   \end{proof}

\section{Proofs of Section~\ref{sec-linear}}
\subsection{Proof of Proposition~\ref{prop-concave} and Corollary~\ref{coro-stratP2}}
\label{appendix-proof-lemma-concave}
\propconcave*
\begin{proof}
	The proof is by induction on~$i$: it is trivial for~$i=0$,
	since~$\calP_0(\loc,v)$ does not depend on~$v$.
	Assume that the result
	holds true for~$\calP_i$, and consider~$\calP_{i+1}$.
	Let~$\loc$ be a state of the automaton, and $(\loc,g,a,r,\loc')$ be its unique outgoing
	transition.
	Let~$(I_j,a)\in\moves(\loc,v_j)$ for $j\in \{1,2\}$. 
	By definition of $\moves$, for $j\in \{1,2\}$  we~then have $v_j+d_j \models g$ for any
	$d_j\in I_j$. 
	We~can then define the set
	$I_{\lambda}=\{\lambda d_1+(1-\lambda)d_2 \mid d_1\in
	I_1,\ d_2\in I_2\}$.
	Moreover, pick any~$d_\lambda \in I_\lambda$: then
	$d_\lambda=\lambda\cdot d_1+(1-\lambda)\cdot d_2$ for
	some~$d_1\in I_1$ and~$d_2\in I_2$.
	Then $v_\lambda+d_\lambda$ can be written as $\lambda\cdot (v_1+d_1) + (1-\lambda)\cdot(v_2+d_2)$.
	Since both $v_1+d_1$ and $v_2+d_2$ satisfy guard~$g$, by convexity of~$g$, we~have that
	$v_\lambda+d_\lambda\models g$.
	This proves that $(I_\lambda,a)\in\moves(\loc,v_\lambda)$.
	Moreover $|I_{\lambda}|= \lambda\cdot|I_1|+ (1-\lambda)\cdot|I_2|$.
	
	Fix~$\varepsilon>0$, and take two intervals~$I_1$ and~$I_2$ such that
	$\min(|I_j|,\inf\{\calP_i(\loc',(v_j+d_j[r\to 0]) \mid d_j\in I_j\}) \geq \calP_i(\loc,v_j)-\epsilon$ for $j\in \{1,2\}$.
	Define~$I_\lambda$ as above.
	Then:
	\begin{xalignat*}1
		\calP_{i+1}(\loc,v_\lambda) &= \sup_{(I,a)\in \moves(\loc,v_\lambda)} \min(|I|, \inf\{\calP_i(\loc',v') \mid \exists d\in I.\ (\loc,v) \trans[d,a] (\loc',v')\})\\
		&\geq
		\min(|I_\lambda|, \inf\{\calP_i(\loc',v_\lambda') \mid \exists d_\lambda\in I_\lambda.\ (\loc,v_\lambda) \trans[d_\lambda,a] (\loc',v_\lambda')\}) \\
		\noalign{\hfill(because the supremum over all moves is larger than or\nopagebreak}
		\noalign{\hfill equal to the value for the particular move~$(I_\lambda,a)$)\pagebreak[1]}
		\noalign{\pagebreak[1]}
		&= \min(|I_\lambda|, \inf\{\calP_i(\loc',(v_\lambda+d_\lambda)[r\to 0])\mid d_\lambda\in I_\lambda\}) \\
		\noalign{\hfill(by expanding the effect of transition $(\loc,g,a,r,\loc')$}
		&= \min(|I_\lambda|, \inf\{\calP_i(\loc',(\lambda\cdot(v_1+d_1)+(1-\lambda)\cdot (v_2+d_2))[r\to 0])\\
		\noalign{\hfill$\mid d_1\in I_1, d_2\in I_2\})$}
		\noalign{\hfill(by defintion of~$I_\lambda$)}
		\noalign{\pagebreak[1]}
		&= \min(|I_\lambda|, \inf\{\calP_i(\loc',(\lambda\cdot((v_1+d_1)[r\to 0])+\\
		\noalign{\hfill$(1-\lambda)\cdot ((v_2+d_2))[r\to 0]))
			\mid d_1\in I_1, d_2\in I_2\})$}
		\noalign{\hfill(by linearity of projection)}
		&\geq \min(|I_\lambda|, \inf\{\lambda\cdot(\calP_i(\loc',(v_1+d_1)[r\to 0])\\
		\noalign{\hfill$ + (1-\lambda)\cdot\calP_i(\loc',(v_2+d_2)[r\to 0]))
			\mid d_1\in I_1, d_2\in I_2\})$} 
		\noalign{\hfill (by induction hypothesis)}
		&= \min(\lambda\cdot|I_1|+ (1-\lambda)|I_2|, \lambda\cdot\inf\{\calP_i(\loc',(v_1+d_1)[r\to 0]) \mid d_1\in I_1\} \\
		\noalign{\hfill$+(1-\lambda)\cdot\inf\{\calP_i(\loc',(v_2+d_2)[r\to 0]) \mid d_2\in I_2\})$}
		%
		\noalign{\pagebreak[1]}
		&\geq \lambda\cdot \min(|I_1|,\inf\{\calP_i(\loc',(v_1+d_1[r\to 0]) \mid d_1\in I_1\}) \\
		\noalign{\hfill$ +(1-\lambda)\cdot\min(|I_2|,\inf\{\calP_i(\loc',(v_2+d_2)[r\to 0]) \mid d_2\in I_2\}))$}
		\noalign{\hfill(as $\min(a+b,a'+b') \geq \min(a,a') + \min(b,b')$)}
		&\geq \lambda\cdot\calP_{i+1}(\loc,v_1) + (1-\lambda)\cdot\calP_{i+1}(\loc,v_2) -\varepsilon.
	\end{xalignat*}
	Since this holds for any~$\varepsilon>0$, our result follows.
	%
\end{proof}

\corostratPII*

\begin{proof}
  The fact that $\calP_i(\loc,v_{\lambda}) \geq
  \min\{\calP_i(\loc,v),\calP_i(\loc,v')\}$ is a direct consequence of
  Proposition~\ref{prop-concave}.
  
  Additionally, we~have
  \begin{multline*}
\inf \{\calP_i(\loc',v')\mid
\exists d\in [\alpha,\beta].\ (\loc,v)\trans[d,a](\loc',v')\}
= \\
\inf \{\calP_i(\loc',v')\mid
  \exists \lambda\in[0,1].\ v'=\lambda \cdot v'_\alpha+(1-\lambda) \cdot v'_\beta\}
  \end{multline*}
  because $(v+(\lambda\alpha+(1-\lambda)\beta))[r\to 0] =
  \lambda\bigl( (v+\alpha)[r\to 0]\bigr) +
  (1-\lambda)\bigl( (v+\beta)[r\to 0]\bigr)$.
  The second claim follows. 
\end{proof}

\subsection{Proof of lemma~\ref{lemma-linperm}}
\label{appendix-proof-lemma-linperm}

For this proof, we use a more precise definition of piecewise-affine functions:

\begin{definition}\label{def-paf}
  Let~$n\in\bbN$.
  An $n$-dimensional piecewise-affine function is a mapping
  $f\colon \bbR+^n \to \bbRbar$ for which
  there exist
  \begin{itemize}
  \item a finite family of $n$-dimensional linear
    functions~$\Phi=(\phi_k)_{1\leq k\leq m}$, and a finite family of finite
    partitions~$P=(P_k)_{1\leq k\leq m}$ of~$\bbR$; these define the
    following partition of~$\bbR+^n$ into convex
    polyhedra (some of which may be empty):
    \[
    \polyh{\Phi,P} = \{\polyh{\Phi,b} \mid b=(b_k)_{1\leq k\leq m}
    \text{ s.t. for all } 1\leq k\leq m,\ b_k\in P_k\}.
    \]
  \item for each convex polyhedron~$h$ of $\polyh{\Phi,P}$,
    an affine function~$f^h$, which we write as
    $f^h(v) = F_0^h + \sum_{1\leq k\leq n} F_k^h\cdot v_k$;
  \end{itemize}
  s.t. for all~$v\in\bbR+^n$, $f(v)=f^h(v)$ for the unique
  cell~$h$ of~$\polyh{\Phi,P}$ containing~$v$.



\end{definition}

\begin{example}
We consider the 2-dimensional affine function~$f$ displayed on
Fig.~\ref{fig-2daff}. Its underlying partition can be defined using
two linear functions:
\begin{itemize}
\item $\phi_1\colon (x,y) \mapsto y$, associated with the partition
  $P_1=\{(-\infty; 1], (1; +\infty)\}$;
\item $\phi_2\colon (x,y) \mapsto x-1$ associated with
  $P_2=\{(-\infty; 0], (0;1], (1; +\infty)\}$.
\end{itemize}

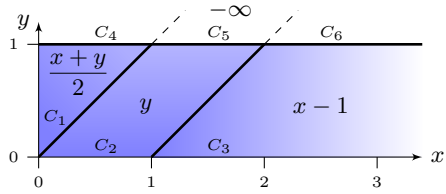
\begin{figure}[h]
  \centering
  \begin{tikzpicture}[scale=1.5]
    \shade[shading=axis,bottom color=blue!50!white,top color=blue!32!white,shading angle=-45] (0,0) |- (1,1) -- cycle;
    \shade[shading=axis,bottom color=blue!50!white,top color=blue!0!white,shading angle=-90] (1,0) -- (1,1) -- (3.4,1) -- (3.4,0) -- cycle;
    \shade[shading=axis,bottom color=blue!50!white,top color=blue!29!white,shading angle=0] (0,0) -- (1,1) -- (2,1) -- (1,0) -- cycle;
    \path[use as bounding box] (0,1.2) -- (0,-.1);
    \draw[latex'-latex'] (0,1.2) node[left] {$y$} |- (3.4,0) node[right] {$x$};
    \draw (0,1) -- +(-.1,0) node[left] {$\scriptstyle 1$};
    \draw (0,0) -- +(-.1,0) node[left] {$\scriptstyle 0$};
    \draw (0,0) -- +(0,-.1) node[below] {$\scriptstyle 0$};
    \draw (1,0) -- +(0,-.1) node[below] {$\scriptstyle 1$};
    \draw (2,0) -- +(0,-.1) node[below] {$\scriptstyle 2$};
    \draw (3,0) -- +(0,-.1) node[below] {$\scriptstyle 3$};
    \draw[line width=1pt] (0,1) -- (3.4,1);
    \draw[line width=1pt] (0,0) -- (1,1);
    \draw[dashed] (1,1) -- (1.3,1.3);
    \draw[line width=1pt] (1,0) -- (2,1);
    \draw[dashed] (2,1) -- (2.3,1.3);
    \path(2.5,.45) node {$x-1$};
    \path(.95,.45) node {$y$};
    \path(.35,.75) node {$\dfrac{x+y}2$};
    \path(1.7,1.3) node {$-\infty$};
    \path(.15,.35) node {$\scriptstyle C_1$};
    \path(.6,.1) node {$\scriptstyle C_2$};
    \path(1.6,.1) node {$\scriptstyle C_3$};
    \path(.6,1.1) node {$\scriptstyle C_4$};
    \path(1.6,1.1) node {$\scriptstyle C_5$};
    \path(2.6,1.1) node {$\scriptstyle C_6$};
  \end{tikzpicture}
  \caption{An example of a 2-dimensional piecewise-affine function}
  \label{fig-2daff}
\end{figure}

This defines a partition of $\bbR+^2$ into six cells: on three of
them~(namely~$C_4$, $C_5$ and~$C_6$), for which $\phi_1(x,y)\in (1;+\infty)$,
our piecewise-affine function~$f$ constantly equals~$-\infty$; for the
other three cells:
\begin{itemize}
\item in $C_1=\{(x,y) \mid \phi_1(x,y)\in(-\infty;1] \text{ and } \phi_2(x,y)\in (-\infty;0]\}$, the affine function~$f^{C_1}$ coincides with $(x,y)\mapsto \frac{x+y}2$;
\item in $C_2=\{(x,y) \mid \phi_1(x,y)\in(-\infty;1] \text{ and } \phi_2(x,y)\in (0;1]\}$, the affine function~$f^{C_2}$ coincides with $(x,y)\mapsto y$;
\item in $C_3=\{(x,y) \mid \phi_1(x,y)\in(-\infty;1] \text{ and } \phi_2(x,y)\in (1;+\infty)\}$, the affine function~$f^{C_3}$ coincides with $(x,y)\mapsto x-1$.
\end{itemize}
\end{example}

\bigskip
\noindent
We~now prove Lemma~\ref{lemma-linperm}:

\lemmalinperm*
\begin{proof}
%
%
%
%
  We~assume that $v\mapsto \calP_{i-1}(\loc',v)$ is not
constantly~$-\infty$ (if~it were the case, then also
$\calP_i(\loc,v)=-\infty$ for all~$v$). Similarly, we~assume that
$\moves(\loc,v)$ is non-empty for some~$v$.
Since $v\mapsto \calP_{i-1}(\loc',v)$ is piecewise-affine,
we~can then fix a
polyhedral partition~$\polyh{\Phi,P}$ and, for each cell~$h$ in this
partition, an affine functions $f_h$, such that
$\calP_{i-1}(\loc',v)=f_h(v)$ for the only cell~$h$ containing~$v$.

Our procedure for computing~$\calP_i$ in~$\loc$ consists in listing
the possible pairs of cells defining~$\calP_{i-1}$ in~$\loc'$ where
the left- and right-bounds of the interval to be proposed lie. Our
approach thus consists in listing each such pair of (possibly identical)
cells~$(h_\alpha,h_\beta)$ in the partition defining~$\calP_{i-1}$ in~$\loc'$, and
\begin{itemize}
\item characterizing the set~$S_{(h_\alpha,h_\beta)}$ of all valuations from which those cells can be
  reached by taking the transition from~$\loc$ to~$\loc'$. We~compute
  this polyhedron using quantifier elimination;
\item computing the ranges for~$\alpha$ and~$\beta$ that
  can be played in order to indeed end up respectively in~$h_\alpha$
  and~$h_\beta$. These are intervals~$I_\alpha$ and $I_\beta$, whose
  bounds are expressed as functions of~$v$.
  Computing these bounds may require refining the polyhedron obtained at
  the previous step into several subpolyhedra, in order to
  express them as affine functions of~$v\in
  S_{(h_\alpha,h_\beta)}$;
\item for each subpolyhedra, compute the optimal values for~$\alpha$
  and~$\beta$: following Corollary~\ref{coro-stratP2}, this amounts to
  find values for~$\alpha\in I_\alpha$ and~$\beta\in I_\beta$ that
  maximize the following function
  \[
  \mu\colon (\alpha,\beta) \mapsto \min\{\beta-\alpha;
  \calP_{i-1}(\loc',(v+\alpha)[z\to 0]);
  \calP_{i-1}(\loc',(v+\beta)[z\to 0])  \} . 
  \]
  This is performed by applying our technical Lemma~\ref{lemma-tech};
  it may again require another refinement of the subpolyhedra, and
  returns an affine function for each subpolyhedron.
\end{itemize}
For each pair~$(h_\alpha,h_\beta)$, we~end up with a (partial)
piecewise-affine function, defined on $S_{(h_\alpha,h_\beta)}$,
returning the optimal permissiveness that can be obtained if playing
interval~$[\alpha,\beta]$ such that taking the transition to~$\loc'$
after delay~$\alpha$ (resp.~$\beta$) leads to~$h_\alpha$
(resp.~$h_\beta$).
Our final step to compute~$\calP_i$ in~$\loc$ consists in
taking the maximun of all these partial functions on their (possibly
overlapping) domains.

Notice that all these computations are performed symbolically
w.r.t~$v$: we~manipulate affine functions of~$v$, with conditions
on~$v$ for our computation to be valid. Figure~\ref{fig-procedure-app}
illustrates the main three steps of this procedure.

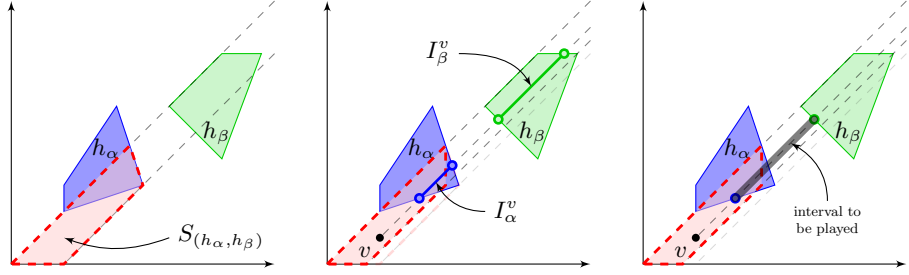
\begin{figure}[ht]
  \centering

  \begin{tikzpicture}[scale=.7]
    \begin{scope}
      \draw[latex'-latex'] (5,0) -| (0,5);
      \draw[vert] (3,3) -- (4,2) -- (4.75,4) -- (4,4) -- cycle;
      \path (3.9,2.5) node {$h_\beta$};
      \draw[bleu] (1,1) -- (1,1.5) -- (2,3) -- (2.5,1.5) -- cycle;
      \path (1.8,2.2) node {$h_\alpha$};
      \draw[dashed,opacity=.5] (5,4) -- (1,0);
      \draw[dashed,opacity=.5] (5,5) -- (0,0);
      \draw[rouge,dashed,line width=1pt,opacity=.5] (0,0) -- (2.25,2.25) --
      (2.5,1.5) -- (1,0) -- cycle;
      \draw[rouge,dashed,line width=1pt,fill=none] (0,0) -- (2.25,2.25) --
      (2.5,1.5) -- (1,0) -- cycle;
      \draw (1,.5) edge[out=-20,in=-160,latex'-]  node[pos=1,right]
            {$S_{(h_\alpha,h_\beta)}$} (3,.5);
    \end{scope}
    \begin{scope}[xshift=6cm]
      \draw[latex'-latex'] (5,0) -| (0,5);
      \draw[vert] (3,3) -- (4,2) -- (4.75,4) -- (4,4) -- cycle;
      \path (3.9,2.5) node {$h_\beta$};
      \draw[bleu] (1,1) -- (1,1.5) -- (2,3) -- (2.5,1.5) -- cycle;
      \path (1.8,2.2) node {$h_\alpha$};
      \draw[dashed,opacity=.5] (5,5) -- (0,0);
      \draw[rouge,dashed,line width=1pt,opacity=.15,fill=none] (0,0) -- (2.25,2.25) --
      (2.5,1.5) -- (1,0) -- cycle;
      \draw[dashed,opacity=.5] (5,4.25) -- (.75,0);
      \draw[dashed,opacity=.15] (5,4) -- (1,0);
      \draw[rouge,dashed,line width=1pt,opacity=.5] (0,0) -- (2.25,2.25) --
      (2.25,1.5) -- (.75,0) -- cycle;
      \draw[rouge,dashed,line width=1pt,fill=none] (0,0) -- (2.25,2.25) --
      (2.25,1.5) -- (.75,0) -- cycle;
      \draw (1,.5) node[inner sep=0pt,fill=black,minimum size=3pt,circle]
      (v) {} node[below left] {$v$};
      \draw[dashed,opacity=.5] (v) -- +(4,4);
      \draw[line width=1pt,bleu] (1.75,1.25) node[inner sep=0pt,bleu,minimum size=3pt,circle] {} -- (2.375,1.875) node[inner sep=0pt,bleu,minimum size=3pt,circle] {} node[midway,coordinate] (b) {};
      \draw[line width=1pt,vert] (3.25,2.75) node[inner sep=0pt,vert,minimum size=3pt,circle] {} -- (4.5,4) node[inner sep=0pt,vert,minimum size=3pt,circle] {}
      node[midway,coordinate] (g) {};
      \draw (b) edge[latex'-,out=-60,in=180]
        node[right,pos=1] {$I^v_{\alpha}$} (3,1);
      \draw (g) edge[latex'-,out=120,in=0]
        node[left,pos=1] {$I^v_{\beta}$} (2.5,4);
    \end{scope}
    \begin{scope}[xshift=12cm]
      \draw[latex'-latex'] (5,0) -| (0,5);
      \draw[vert] (3,3) -- (4,2) -- (4.75,4) -- (4,4) -- cycle;
      \path (3.9,2.5) node {$h_\beta$};
      \draw[bleu] (1,1) -- (1,1.5) -- (2,3) -- (2.5,1.5) -- cycle;
      \path (1.8,2.2) node {$h_\alpha$};
      \draw[dashed,opacity=.5] (5,4.25) -- (.75,0);
      \draw[dashed,opacity=.15] (5,4) -- (1,0);
      \draw[dashed,opacity=.5] (5,5) -- (0,0);
      \draw[rouge,dashed,line width=1pt,opacity=.5] (0,0) -- (2.25,2.25) --
      (2.25,1.5) -- (.75,0) -- cycle;
      \draw[rouge,dashed,line width=1pt,fill=none] (0,0) -- (2.25,2.25) --
      (2.25,1.5) -- (.75,0) -- cycle;

      \draw (1,.5) node[inner sep=0pt,fill=black,minimum size=3pt,circle]
      (v) {} node[below left] {$v$};
      \draw[dashed,opacity=.5] (v) -- +(4,4);
      \draw[line width=1pt,bleu]
      (1.75,1.25) node[inner sep=0pt,bleu,minimum size=3pt,circle] {};
      \draw[line width=1pt,vert] (3.25,2.75) node[inner sep=0pt,vert,minimum size=3pt,circle] {};
      \draw[opacity=.5,black,line cap=round,line width=3pt] (1.75,1.25) -- (3.25,2.75) node[pos=.7,coordinate] (m) {};
      \draw (m) edge[latex'-,out=-30,in=90]  node[pos=1,below,text width=2cm,scale=.6,align=center]
        {interval to be played} (3.5,1.2);
    \end{scope} 
  \end{tikzpicture}
        
  \caption{Three steps of our procedure:
    $S_{(h_\alpha,h_\beta)}$; then compute expressions for
    $I^v_\alpha$ and $I^v_\beta$ (notice that we had to refine
    $S_{(h_\alpha,h_\beta)}$, because the expression for $I^v_\beta$
    would be different for the lower part of $S_{(h_\alpha,h_\beta)}$ since
    it ends of a different facet of~$h_\beta$); 
    finally select best values for~$\alpha$ and~$\beta$.}
  \label{fig-procedure-app}
\end{figure}

We now detail the first three steps. For this, we fix to
cells~$h_\alpha$ and~$h_\beta$ of the partition defining~$\calP_{i-1}$
in~$\loc'$.
Following Definition~\ref{def-paf},
those cells can be characterized by two families, $(b^\alpha_j)_{j}$ and
$(b^\beta_j)_j$, of cells in~$P$, such that $h_\alpha=\bigcap_{1\leq j\leq m}
\phi_j^{-1}(b^\alpha_j)$ and $h_\beta=\bigcap_{1\leq j\leq m}
\phi_j^{-1}(b^\beta_j)$. 

\subsubsection{Computing~$S_{(h_\alpha,h_\beta)}$.}
We~assume that the conjunction of the guard~$g$ and the invariant~$I(\loc)$
can be represented as the conjunction of
one interval constraint~$[L^c,U^c]$ per clock~$c$.
%
The~set~$S_{(h_\alpha,h_\beta)}$ of valuations that can
reach~$h_\alpha$ and~$h_\beta$ with delays~$\alpha\leq \beta$ (and
after taking the transition to~$\loc'$) is defined as follows:
\begin{multline*}
S_{h_\alpha,h_\beta} = \{v\in\bbR+^n\mid \exists 0\leq\alpha\leq\beta.\
  \forall j.\  \phi_j(v+\alpha[z\to 0])\in b^\alpha_j \text{ and } \\
  \phi_j(v+\beta[z\to 0])\in b^\beta_j 
  \text{ and } v+\alpha\in g \text{ and } v+\beta\in g
  \}.
\end{multline*}

Writing $K_j$ for the sum of all coefficients in the linear
function~$\phi_j$, and $b_j^\alpha=[l_j^\alpha,u_j^\alpha]$, condition
$\phi_j(v+\alpha[z\to 0])\in b_j^\alpha$ can be rewritten either as
$\phi_j(v)\in b_j^\alpha$ if~$K_j=0$, or as
\[
\frac1{K_j} (l_j^\alpha-\phi_j(v)) \leq\alpha\leq \frac1{K_j}
(u_j^\alpha-\phi_j(v))
\]
otherwise. The~same applies for~$\beta$.  Using Fourier-Motzkin
quantifier elimination, $S_{h_\alpha,h_\beta}$ can be written as the
conjunction of the following four sets of constraints:%
\begin{itemize}
\item existence of~$\alpha$ is expressed as the conjunction of
  at most $(m+n)\cdot(m+n-1)$ conditions:
  \begin{itemize}
  \item $m\cdot(m-1)$ conditions are of the form $\frac1{K_j}
    (l_j^\alpha-\phi_j(v)) \leq \frac1{K_{j'}}
    (u_{j'}^\alpha-\phi_{j'}(v))$. Notice that in those conditions,
    the coefficients of the linear functions sum up to~zero (we~name
    such linear functions \emph{diagonal} in the sequel);
  \item $mn$ conditions of the form $\frac1{K_j} (l_j^\alpha-\phi_j(v))
    \leq U^c-v(c)$ and $mn$ conditions of the form $L^c-v(c) \leq
    \frac1{K_j} (u_j^\alpha-\phi_j(v))$. Again,
    this gives diagonal linear functions;
  \item $n\cdot (n-1)$ conditions of the form $L^c-v(c)\leq
    U^{c'}-v(c')$; These also give rise to diagonal affine functions;
  \end{itemize}
\item existence of~$\beta$ is expressed similarly. In~particular, it~uses
  the same linear functions, which are all diagonal;
\item that~$\alpha$ is nonnegative is expressed as the conjunction of $\frac1{K_j}
  (u_j^\alpha-\phi_j(v))\geq 0$ for all~$j$ and $U_c-v(c)\geq 0$ for all~$c$;
\item that $\alpha\leq\beta$ is expressed as $(m+n)\cdot (m+n-1)$
  constraints similar to those in the first case. Again, no new linear
  functions are created in this step, compared to the previous ones,
  and only diagonal functions will be used.
\end{itemize}
In~the end we~have $3(m+n)(m+n-1)+(m+n)$ constraints (but defined with
at most $(m+n)^2$ linear functions). Notice that at most $m+n$ of those
linear functions may be non-diagonal, and those non-diagonal functions directly
originate either from the guards or from the partition defining~$\calP_{i-1}$ in~$\loc'$.

\subsubsection{Computing~the range for the bounds~$\alpha$ and~$\beta$.}
In~case $S_{(h_\alpha,h_\beta)}$ is non-empty, we~proceed with computing the 
values for~$\alpha$ and~$\beta$ that indeed lead to~$h_\alpha$ and~$h_\beta$.
For any~$v\in S_{(h_\alpha,h_\beta)}$, the~set~$I^v_\alpha$ (resp.~$I^v_\beta$) of values
for~$\alpha$ (resp.~$\beta$) for which $v+\alpha\models g$ and
$(v+\alpha)[z\to 0]\in h_\alpha$ (resp. $v+\beta\models g$ and
$(v+\beta)[z\to 0]\in h_\beta$) then is an interval: from the conditions above,
these~intervals can be written\footnote{Notice that the bounds of those intervals may be left- and\slash or right open; we only consider closed intervals to not blur the focus of our presentation, but we could handle open intervals easily.} 
$I^v_\alpha= [ D^v_\alpha, E^v_\alpha ]$ with
\begin{xalignat*}1
D^v_\alpha &= \max\left(\Bigl\{\frac1{K_j} (l^\alpha_j-\phi_j(v)) \mid 1\leq j\leq m\Bigr\} \cup
  \{L_c-v(C) \mid c\in\clocks\}\right)
\\
E^v_\alpha &= \min\left(\Bigl\{\frac1{K_j} (u^\alpha_j-\phi_j(v)) \mid 1\leq j\leq m\Bigr\} \cup
  \{U_c-v(C) \mid c\in\clocks\}\right)
\end{xalignat*}
(and similarly for~$I_\beta$).  In~order to have affine expressions
for the bounds of $I^v_\alpha$ and $I^v_\beta$, we~refine
$S_{(h_\alpha,h_\beta)}$ into cells on which one of the affine
functions in the expressions of~$D^v_\alpha$ (resp.~$E^v_\alpha$)
realizes the maximum (resp. minimum). This refinement is obtained by
expressing the fact that the selected affine function is indeed larger
than (resp. smaller than) all other functions. This may
refine~$S_{(h_\alpha,h_\beta)}$ into at most $(m+n)^4$ cells, defined
with diagonal linear functions already introduced at the previous step.

\subsubsection{Computing the optimal values for~$\alpha$ and~$\beta$.}
We~let $D=\{(\alpha,\beta) \mid \alpha\in I^v_\alpha,\ \beta\in
I^v_\beta,\ \alpha\leq\beta\}$. It~remains to find the optimal choices
for~$\alpha$ and~$\beta$, i.e., the values that maximize
\[
\min\{\beta-\alpha;
  \inf_{\gamma\in[\alpha;\beta]}\{\calP_{i-1}(\loc',(v+\gamma)[r\to 0])\}\}
\]
over~$D$.
Thanks to Corollary~\ref{coro-stratP2}, this  amounts to maximizing
\[
\mu\colon (\alpha,\beta) \mapsto \min\{\beta-\alpha;
  \calP_{i-1}(\loc',(v+\alpha)[z\to 0]);
  \calP_{i-1}(\loc',(v+\beta)[z\to 0])  \}
\]
over that~set.
Since $(v+\alpha)[z\to 0]\in h_\alpha$, we~have
$\calP_{i-1}(\loc',(v+\alpha)[z\to 0]) = f_{h_\alpha}((v+\alpha)[z\to
  0])$.  Function~$f_{h_\alpha}$ is an n-dimensional affine function;
writing~$F_{h_\alpha}^z$ for the sum of the coefficients of the clocks
that are not reset in~$z$ (i.e.,
$F^z_{h_\alpha}=f_{h_\alpha}(\bfone)-f_{h_\alpha}(\bbone_z)$),
we~have $f_{h_\alpha}((v+\alpha)[z\to 0] = F^z_{h_\alpha}\cdot\alpha +
f_{h_\alpha}(v[z\to 0])$. Similarly for~$\beta$. We~then~have
\[
\mu(\alpha,\beta) = \min\{\beta-\alpha;
F^z_{h_\alpha}\cdot\alpha + f_{h_\alpha}(v[z\to 0]);
F^z_{h_\beta}\cdot\beta + f_{h_\beta}(v[z\to 0])
\},
\]
which we want to minimize over~$D$. In~case~$I^v_\beta$ is unbounded
(which may occur when all upper bounds~$u^\beta_j$ and~$U_c$ equal
$+\infty$), by~Lemma~\ref{lemma-infini} we~get that $f_{h_\beta}$ is
constant, and we can choose $\beta=+\infty$. It~remains to maximize
$\min\{F^z_{h_\alpha}\cdot\alpha + f_{h_\alpha}(v[z\to 0]),
f_{h_\beta}(v[z\to 0])\}$ when~$\alpha$ ranges
over~$I^v_\alpha$. Again, if~$I^v_\alpha$ is unbounded, $f_{h_\alpha}$
is constant, and~$\alpha$ can be chosen arbitrarily in~$I^v_\alpha$;
otherwise, the maximum is obtained at one of the bounds
of~$I^v_\alpha$, depending on the sign of~$F^z_{h_\alpha}$.

Now, in case~$I^v_\alpha$ and~$I^v_\beta$ are bounded, we~apply
Lemma~\ref{lemma-technique} and~directly get the optimal
solution. This may again require refining the polyhedron being
considered into at most~$13$ subpolyhedra, since there may be up to
$13$ different cases for minimizing~$\mu(\alpha,\beta)$
(see~Appendix~\ref{appendix-proof-tech}).  We~the get the optimal values for~$\alpha$
and~$\beta$, as well as the value of~$\mu$ at that maximal point,
depending on the signs of $F^z_{h_\alpha}$ and $F^z_{h_\beta}$. It~can
be checked that both the optimal choices for~$\alpha$ and~$\beta$, as
well as the resulting permissiveness function, are linear functions
of~$v$: indeed, in~our instance of the problem of
Lemma~\ref{lemma-technique}, $a$ and~$c$ are constant, while $b$
and~$d$, and $m_x$, $M_x$, $m_y$ and~$M_y$ are affine functions
of~$v$; the~latter may only be multiplied by constants, and\slash or
added with one another.

The~coefficients of those affine functions can be computed from those
of~$f_{h_\alpha}$ and $f_{h_\beta}$, and from those of
functions~$(\phi_j)_j$ defining the partiiton of the piecewise-linear
permissiveness function~$\calP_i(\loc',v)$: in~the worst case,
the~numerators are multipled by the sum of all coefficients, and the
denominators may be multiplied by the product of two sums of
coefficients. In~any~case, the~space needed to store one such function
(assuming binary encoding) is at most linear in the space needed to
store~$\calP_{i-1}$.  The concludes the third step of our computation.

\subsubsection{Finalizing the computation of~$\calP_i$ in~$\loc$.}
We~now have a collection of affine functions (at~most $13m^2\cdot
(m+n)^4$), associated with a polyhedron on which they give a candidate
expression for~$\calP_i$. The~polyhedra may overlap, as for each
valuation we considered $m^2$ possible cells in which the valuation
may end up. We~thus have to refine one last time the partition we
obtained, by considering subpolyhedra where one of the $m^2$ candidate
functions is larger than the other ones. This may further refine each
cell into~$m^2$ subpolyhedra, defined with up to $m^2$ new linear functions.

In the end, this proves that the function $\calP_i$ in~$\loc$ is
piecewise affine, and that it can be computed from~$\calP_{i-1}$ in
time~$O(m^4\cdot(m+n)^4)$.  The~partition defining~$\calP_i$ in~$\loc$ may have up
to $O(m^4\cdot (m+n)^4)$ cells, defined with at most $O((m+n)^2)$ linear functions.
The~coefficients of the affine functions defining~$\calP_i$ are
polynomials in the coefficients of the affine functions
defining~$\calP_{i-1}$.
\end{proof}

\section{Example of computation of permissiveness}
\label{app-example-long}
\begin{example}
  We slightly modify the automaton of Fig.~\ref{fig-experm}, by changing
the guard on the first transition as displayed on
Fig.~\ref{fig-experm2-app}.  We develop the computation of the
permissiveness function for this automaton.
\begin{figure}[!ht]
  \centering
  \begin{tikzpicture}
    \begin{scope}
      \draw (0,0) node[rond,bleu] (a) {$\loc_0$};
      \draw (4,0) node[rond,rouge] (b) {$\loc_1$};
      \draw (8,0) node[rond,vert] (c) {$\loc_f$};
      \draw (a) edge[draw,-latex'] node[above] {$\begin{array}{>{\scriptstyle}c}
          0\leq y\leq 1
        \end{array}$} node[below] {$\scriptstyle y:=0$} (b);
      \draw (b) edge[draw,-latex'] node[above] {$\begin{array}{>{\scriptstyle}c}
          1\leq x\leq 2 \\ 0\leq y\leq 1
        \end{array}$} (c);
    \end{scope}
  \end{tikzpicture}
  \caption{Automaton of Fig.~\ref{fig-experm} where the guard on the
    first transition has been slightly extended}
  \label{fig-experm2-app}
\end{figure}
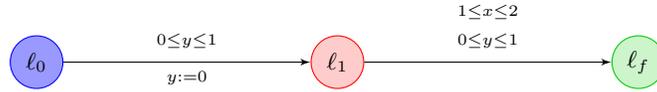

Obviously, function $\calP(\loc_1)=\calP_1(\loc_1)$ is
unchanged. We~detail the computation of~$\calP_2(\loc_0)$. Following
the proof of Lemma~\ref{lemma-linperm}, we~list the pairs of possible
cells where the automaton may enter~$\loc_1$ after delays~$\alpha$
and~$\beta$: since the transition to~$\loc_1$ resets~$y$, we have two
possible cells, namely\footnote{For convenience in this 2-clock
  example, we may write valuations either as~$v$ or as pairs~$(x,y)$,
  depending on the situation.} $C_0=\{(x,0) \mid 0\leq x\leq 1\}$ and
$C_1=\{(x,0) \mid 1<x\leq 2\}$. Hence we have four possible situations to
consider:
\begin{enumerate}
\item both $(v+\alpha)[y\to 0]$ and $(v+\beta)[y\to 0]$ in~$C_0$;
\item both $(v+\alpha)[y\to 0]$ and $(v+\beta)[y\to 0]$ in~$C_1$;
\item $(v+\alpha)[y\to 0]\in C_0$ and $(v+\beta)[y\to 0]\in C_1$;
\item $(v+\alpha)[y\to 0]\in C_1$ and $(v+\beta)[y\to 0]\in C_0$.
\end{enumerate}

For each pair, we~begin with computing the set of valuations~$v$ for
which there are values~$0\leq \alpha\leq \beta$ satisfying the conditions:
\begin{enumerate}
\item having $(v+\alpha)[y\to 0]$ and $(v+\beta)[y\to 0]$ in~$C_0$ can
  be written as
  \begin{multline*}
  \exists \alpha\leq\beta.\
  0\leq v(y)+\alpha\leq 1 \et 0\leq v(y)+\beta\leq 1 \et {}\\
  0\leq v(x)+\alpha\leq 1 \et 0\leq v(x)+\beta\leq 1.
  \end{multline*}
  The~constraints on~$y$ come from the guard of the transition, while
  those on~$x$ correspond to having the target valuations in~$C_0$.
  In~this simple case, quantifier elimination returns $\{(x,y)
  \mid 0\leq x\leq 1 \et 0\leq y\leq 1\}$.
\item having $(v+\alpha)[y\to 0]$ and $(v+\beta)[y\to 0]$ in~$C_1$
  translates to
  \begin{multline*}
  \exists \alpha\leq\beta.\
  0\leq v(y)+\alpha\leq 1 \et 0\leq v(y)+\beta\leq 1 \et \\
  1< v(x)+\alpha\leq 2 \et 1< v(x)+\beta\leq 2.
  \end{multline*}
  This results in $\{(x,y) \mid 0\leq x\leq 2 \et 0\leq y\leq 1 \et
  y\leq x\}$.
\item the case where $(v+\alpha)[y\to 0]\in C_0$ and $(v+\beta)[y\to 0]\in C_1$
  writes
  \begin{multline*}
  \exists \alpha\leq\beta.\
  0\leq v(y)+\alpha\leq 1 \et 0\leq v(y)+\beta\leq 1 \et \\
  0\leq v(x)+\alpha\leq 1 \et 1< v(x)+\beta\leq 2.
  \end{multline*}
  This corresponds to $\{(x,y) \mid 0\leq x\leq 1 \et 0\leq
  y\leq 1 \et y\leq x\}$.
\item Finally, the situation where
  $(v+\alpha)[y\to 0]\in C_1$ and $(v+\beta)[y\to 0]\in C_0$ translates as
  \begin{multline*}
  \exists \alpha\leq\beta.\
  0\leq v(y)+\alpha\leq 1 \et 0\leq v(y)+\beta\leq 1 \et \\
  1< v(x)+\alpha\leq 2 \et 0\leq v(x)+\beta\leq 1.  
  \end{multline*}
  This in particular requires $1-v(x)<\alpha$ and $\beta\leq 1-v(x)$,
  which are incompatible with the condition $\alpha\leq\beta$. Hence
  this case is empty.
\end{enumerate}

We~now compute the intervals of possible values for~$\alpha$
and~$\beta$: this just amounts to writing the conditions for having
$v+\alpha$ satisfy the guard and $(v+\alpha)[y\to 0]$ belong to the
target cell (the~computation is identical for~$\beta$):
\begin{itemize}
\item having $(v+\alpha)[y\to 0]$ end up in~$C_0$ requires
  $\alpha \in [0;1-v(x)]\cap [0;1-v(y)]$;
\item having $(v+\alpha)[y\to 0]$ end up in~$C_1$ requires
  $\alpha \in (1-v(x);2-v(x)]\cap [0;1-v(y)]$.
\end{itemize}

We~end up with the following situations:
\begin{enumerate}
\item having both $(v+\alpha)[y:=0]$ and $(v+\beta)[y:=0]$ in~$C_0$
  can be performed from $\{(x,y) \mid 0\leq x\leq 1 \et 0\leq y\leq
  1\}$; from that zone:
  \begin{enumerate}
  \item\label{pt1a} if $x\leq y$, we~have $I_\alpha^v=I_\beta^v=[0;1-y]$;
  \item\label{pt1b} if $x > y$, we~have $I_\alpha^v=I_\beta^v=[0;1-x]$.
  \end{enumerate}
\item having both $(v+\alpha)[y:=0]$ and $(v+\beta)[y:=0]$ in~$C_1$
  can be performed from $\{(x,y) \mid 0\leq x\leq 2 \et 0\leq y\leq 1 \et
  y\leq x\}$; from that zone:
  \begin{enumerate}
  \item\label{pt2a} if $x\leq 1$ and $y<x$, we~have $I_\alpha^v=I_\beta^v=(1-x;1-y]$;
  \item\label{pt2b} if $1\leq x\leq 1+y$, we~have $I_\alpha^v=I_\beta^v=[0;1-y]$;
  \item\label{pt2c} if $1+y\leq x\leq 2$, we~have $I_\alpha^v=I_\beta^v=[0;2-x]$.
  \end{enumerate}
\item\label{pt3} having $(v+\alpha)[y:=0]\in C_0$ and $(v+\beta)[y:=0]\in C_1$
  can be performed from $\{(x,y) \mid 0\leq x\leq 1 \et 0\leq
  y\leq 1 \et y\leq x\}$. We~then have
  $I_\alpha^v=[0;1-x]$ and $I_\beta^v=(1-x,1-y]$.
\end{enumerate}

We now have to compute the optimal values of~$\alpha$ and~$\beta$ in
each of these six situations:
\begin{itemize}
\item[\eqref{pt1a}] for the first situation, we~have to maximize
  $(\alpha,\beta)\mapsto \min\{\beta-\alpha, x+\alpha,x+\beta\}$ over
  $\{(\alpha,\beta) \mid \alpha\in [0;1-y], \beta\in [0;1-y],
  \alpha\leq\beta\}$.
  This corresponds to case ``$a\geq 0$ and $c\geq 0$'' of
  Lemma~\ref{lemma-tech}. We~get:
  \begin{itemize}
  \item if $\frac{1-y-x}2\leq 0$, the optimal interval for the player
    is $[0;1-y]$, yielding permissiveness $1-y$;
  \item if $0\leq \frac{1-y-x}2$, the optimal interval is
    $[\frac{1-y-x}2;1-y]$, with permissiveness $\frac{1+x-y}2$.
  \end{itemize}

\item[\eqref{pt1b}] in the second situation, we maximize
  $(\alpha,\beta)\mapsto \min\{\beta-\alpha, x+\alpha,x+\beta\}$ over
  $\{(\alpha,\beta) \mid \alpha\in [0;1-x], \beta\in [0;1-x],
  \alpha\leq\beta\}$. The situation is the same as above, and we get:
  \begin{itemize}
  \item if $\frac12-x\leq 0$, the optimal interval for the player
    is $[0;1-x]$, yielding permissiveness $1-x$;
  \item if $0\leq \frac12-x$, the optimal interval is
    $[\frac12-x;1-x]$, with permissiveness $\frac12$.
  \end{itemize}

\item[\eqref{pt2a}] in the third situation, we maximize
  $(\alpha,\beta)\mapsto
  \min\{\beta-\alpha, 2-(x+\alpha),2-(x+\beta)\}$ over $\{(\alpha,\beta)
  \mid \alpha\in (1-x,1-y], \beta\in (1-x;1-y],
  \alpha\leq\beta\}$. We~apply Lemma~\ref{lemma-tech}, with $a\leq
  0$ and $c\leq 0$:
  \begin{itemize}
  \item the first condition corresponds to $x\leq \frac12+y$; there
    the maximal point is $\min\{x-y,1\}$, i.e. $x-y$, and is reached
    at $(1-x,1-y)$. Since the bound at $1-x$ is strict,
    we take $1-x+\varepsilon$ instead of~$1-x$, for some arbitrarily small~$\varepsilon>0$.
  \item the second condition is $x\geq \frac12+y$, for which the
    maximal value is~$\frac12$ is reached at $(1-x,\frac32-x)$. Again, we~have to take
    $1-x+\varepsilon$ instead of~$1-x$.
  \end{itemize}


\item[\eqref{pt2b}] in this case, we maximize the same function over
  $\{(\alpha,\beta) \mid \alpha\in [0,1-y], \beta\in [0;1-y],
  \alpha\leq\beta\}$  over the zone ($1\leq x\leq
  1+y\leq 2$):
  \begin{itemize}
  \item the first condition is $y\geq \frac x2$: in that zone, the maximal
    point is $1-y$, reached for $(0,1-y)$;
  \item the second condition is $y\leq \frac x2$,
    and the maximal value $1-\frac x2$ is reached at $(0,1-\frac x2)$.
  \end{itemize}
  

\item[\eqref{pt2c}] we maximize the same function over
  $\{(\alpha,\beta) \mid \alpha\in [0,2-x], \beta\in [0;2-x],
  \alpha\leq\beta\}$. Again, the second condition holds, and the
  maximal value is $1-\frac x2$, reached at $(0,1-\frac x2)$.


\item[\eqref{pt3}] finally, we~have to maximize $(\alpha,\beta)
  \mapsto \min\{\beta-\alpha, x+\alpha, 2-x-\beta\}$ over
  $\{(\alpha,\beta) \mid \alpha\in [0,1-x], \beta\in (1-x;1-y],
    \alpha\leq\beta\}$. Hence we are in case $a\geq 0$ and $c\leq 0$
    of Lemma~\ref{lemma-tech}. In no cases can the first and second
    condition hold. Then:
    \begin{itemize}
    \item when $y\geq 1-x$ and $y\geq \frac x2$, then the third
      condition holds, and the maximal value $1-y$ is reached at $(0,1-y)$.
      \medskip

      Now, let $T_x=\frac23-x$ and $T_y=\frac43-x$.
    \item the fourth condition rewrites as $y\geq x-\frac13$ (and the
      complement of the previous condition). For those points, the
      maximal value is $\frac{1+x-y}2$, reached at
      $(\frac{1-x-y}2,1-y)$.
    \item the fifth condition is $x\geq \frac23$ (and the complement
      of the condition above). There the maximal point $1-\frac x2$ is
      reached for $(0,1-\frac x2)$.
    \item for the remaining points: we have $ad=2-x\geq -x=bc$, and
      neither $T_y\leq m_y$ nor $T_x\geq M_x$ hold, so that the
      optimal point is $2/3$, corresponding to $(\frac23-x,
      \frac43-x)$.
    \end{itemize}
\end{itemize}

By superimposing those results and taking the maxima on cells where
several solutions have been computed, we~get the global permissiveness
function depicted on Fig.~\ref{fig-permexlong-app}.

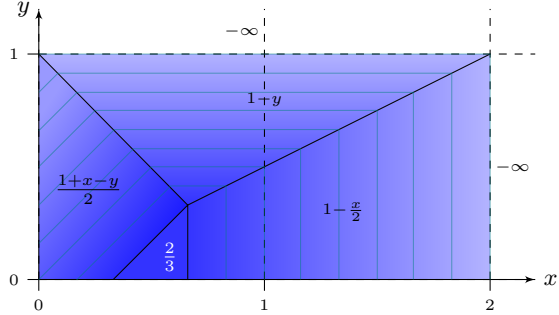
\begin{figure}[!ht]
  \centering
  \begin{tikzpicture}
    \begin{scope}[xshift=-1cm,yshift=-3cm,scale=3]
      \fill[fill=blue!80!white] (.33,0) -| (.66,.33) -- cycle;
      \shade[left color=blue!80!white, right color=blue!30!white]
        (.66,0) -- (.66,.33) -- (2,1) |- cycle;
      \shade[shading=axis,bottom color=blue!90!white,top color=blue!40!white,shading angle=45]
      (0,0) -- (.33,0) -- (.66,.33) -- (0,1) -- cycle;
      \shade[bottom color=blue!80!white,top color=blue!30!white]
      (0,1) -- (.66,.33) -- (2,1) -- cycle;

      \path[use as bounding box] (0,1.2) -- (0,-.1);

      \draw[latex'-latex'] (2.2,0) node[right] {$x$} -| (0,1.2)
        node[left] {$y$};
      \begin{scope}
        \path[clip] (-.5,-.5) rectangle (2.2,1.2);
      \foreach \x in {0,1,2}
           {\draw (\x,0) -- +(0,-.05) node[below] {$\scriptstyle \x$};
            \draw (0,\x) -- +(-.05,0) node[left] {$\scriptstyle \x$};
            \draw[dashed] (\x,0) -- +(0,1.2);
            \draw[dashed] (0,\x) -- +(2.2,0);
           }
      \draw (.33,0) -- (.66,.33) -- (.66,0);
      \draw (0,1) -- (.66,.33) -- (2,1);
      \path[clip] (0,0) rectangle (2.2,1.2);
      \foreach \x in {.83,1,1.16,1.33,1.5,1.66,1.83,2} 
        {\draw[blue!50!green,opacity=.5] (\x,0) -- (\x,\x/2) -- (1-\x/2,\x/2) -- (0,\x-1);}
      \end{scope}
      
      \draw (.58,.1) node[color=white] {$\frac23$};
      \path (1.35,.3) node {$\scriptstyle 1-\frac x2$};
      \draw (1,.8) node {$\scriptstyle 1-y$};
      \draw (.22,.4) node {$\frac{1+x-y}2$};
      \draw (.9,1.1) node {$\scriptstyle -\infty$};
      \draw (2.1,.5) node {$\scriptstyle -\infty$};
    \end{scope}
  \end{tikzpicture}
  \caption{A linear timed automaton and its permissiveness at~$\loc_0$}
  \label{fig-permexlong-app}
\end{figure}

\end{example}

\section{Proof of Lemma~\ref{lemma-technique}}
\label{appendix-proof-tech}

\lemmatech*


\begin{proof}
We~write~$h$ for the function $(\alpha,\beta) \mapsto \beta-\alpha$.
We~assume that~$D$ is non-empty (i.e., $m_\alpha\leq M_\beta$).  We split the
  proof into four cases, depending on the signs of~$a$ and~$c$.

\paragraph{\llap{$\blacktriangleright$\ }When $a\leq 0$ and $c\geq 0$,}
then all three functions defining~$\mu$ are maximized when $\alpha=m_\alpha$ and
$\beta=M_\beta$. It~follows that the maximal value of~$\mu$ over~$D$ is
$\min\{M_\beta-m_\alpha, f(m_\alpha), g(M_\beta)\}$, and is reached at $(m_\alpha,M_\beta)$.

\begin{figure}[ht]
\centering
\begin{tikzpicture}
\begin{scope}[scale=1.3]
\fill[black!10!white] (2,0) -- (.8,.8) -- (0,.7) |- cycle;
\fill[red!50!white] (0,.7) -- (.8,.8) -- (0,1.8) |- cycle;
\fill[blue!50!white] (0,1.8) -- (.8,.8) -- (2,0) -- (2,2) -| cycle;
\draw[latex'-latex'] (2.2,0) node[right] {$\alpha$} -| (0,2.2) node[left] {$\beta$};
\draw (0,1.8) -- (.8,.8) -- (0,.7);
\draw (.8,.8) -- (2,0);
\everymath{\scriptstyle}
\path (.3,1) node {$g$}; \draw[->] (.1,.9) -- +(90:4mm);
\path (1.5,.8) node {$f$}; \draw[->] (1.5,1.7) -- +(180:4mm);
\path (.2,.4) node {$h$}; \draw[->] (.7,.2) -- +(135:4mm);
\fill[black,opacity=.2] (.8,.6) -- (1.4,1.2) |- (.5,1.5) |- cycle;
\draw[dotted] (.8,.6) -- (1.4,1.2) |- (.5,1.5) |- cycle;
\path (1.3,1.3) node {$D$};
\end{scope}
\begin{scope}[xshift=4cm,scale=1.3]
\fill[blue!50!white] (0,.5) -- (.8,.8) -- (2,2) -| cycle;
\fill[red!50!white] (2,1) -- (.8,.8) -- (2,2) -- cycle;
\fill[black!10!white] (0,.5) -- (.8,.8) -- (2,1) -- (2,0) -| cycle;
\draw[latex'-latex'] (2.2,0) node[right] {$\alpha$} -| (0,2.2) node[left] {$\beta$};
\draw (0,.5) -- (.8,.8) -- (2,1);
\draw (.8,.8) -- (2,2);
\everymath{\scriptstyle}
\path (.3,1.2) node {$f$}; \draw[->] (.7,1.8) -- +(0:4mm);
\path (1.8,1.4) node {$g$}; \draw[->] (1.6,1) -- +(90:4mm);
\path (1.5,.4) node {$h$}; \draw[->] (1.2,.1) -- +(135:4mm);
\fill[black,opacity=.2] (.8,.6) -- (1.4,1.2) |- (.5,1.5) |- cycle;
\draw[dotted] (.8,.6) -- (1.4,1.2) |- (.5,1.5) |- cycle;
\path (1.1,1.3) node {$D$};
\end{scope}
\begin{scope}[xshift=8cm,scale=1.3]
\fill[blue!50!white] (0,0.3) -- (.8,.8) -- (0,1.6) -- cycle;
\fill[red!50!white] (0,1.6) -- (.8,.8) -- (2,1.2) |- (0,2) -- cycle;
\fill[black!10!white] (0,.3) -- (.8,.8) -- (2,1.2) |- (0,0) -- cycle;
\draw[latex'-latex'] (2.2,0) node[right] {$\alpha$} -| (0,2.2) node[left] {$\beta$};
\draw (0,.3) -- (.8,.8) -- (2,1.2);
\draw (.8,.8) -- (0,1.6);
\everymath{\scriptstyle}
\path (.2,1) node {$f$}; \draw[->] (.1,.8) -- +(0:4mm);
\path (1.6,.5) node {$h$}; \draw[->] (1.2,.2) -- +(135:4mm);
\path (1.6,1.7) node {$g$}; \draw[->] (.6,1.8) -- +(-90:4mm);
\begin{scope}[xshift=4mm]
\fill[black,opacity=.2] (.8,.6) -- (1.4,1.2) |- (.5,1.5) |- cycle;
\draw[dotted] (.8,.6) -- (1.4,1.2) |- (.5,1.5) |- cycle;
\path (1.1,1.3) node {$D$};
\end{scope}
\end{scope}
\end{tikzpicture}
\caption{Four cases for the proof of Lemma~\ref{lemma-technique}: when
  $a\leq 0$ and $c\geq 0$ (left); when $a\geq 0$ and $c\geq 0$, and
  symmetrically, $a\leq 0$ and $c\leq 0$ (center); when $a\geq 0$ and
  $c\leq 0$ (right).  In~each picture, the state space
  $\protect\bbR+^2$ is divided into three cells, depending on which
  function is minimal among~$f$, $g$ and~$h$. In~each cell, we~also
  indicate the direction in which those functions increase.}
\label{fig-lmtech}
\end{figure}
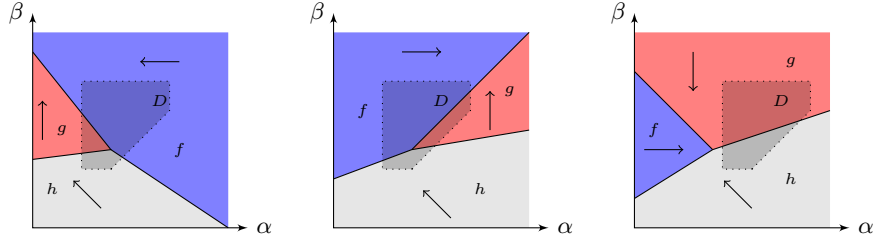

\paragraph{\llap{$\blacktriangleright$\ }When $a\geq 0$ and $c\geq 0$,}
then for two points~$(\alpha,\beta)$ and~$(\alpha',\beta')$ in~$D$
with~$\beta\leq \beta'$, it holds $(\alpha,\beta')\in D$ and we~have
$\mu(\alpha,\beta) \leq \mu(\alpha,\beta')$. Hence $\mu$ is maximized
over~$D$ at a point where $\beta=M_\beta$. It~remains to
maximize~$\alpha\mapsto \mu(\alpha,M_\beta)$ over $\{\alpha\in\bbR\mid
(\alpha,M_\beta)\in D\}$. Over~$\bbR$, this function is maximized for
$\alpha_0=\frac{M_\beta-b}{a+1}$, where
$\mu(\alpha_0,M_\beta)=\min\{\frac{a\cdot M_\beta+b}{a+1}, c\cdot
M_\beta+d\}$.  If $(\alpha_0,M_\beta)\in D$, this is the maximum
of~$\mu$ over~$D$, otherwise the maximum is reached on the border of
$\{\alpha\in\bbR\mid {(\alpha,M_\beta)\in D}\}$,
i.e. for~$\alpha=m_\alpha$ or $\alpha=\min\{M_\beta, M_\alpha\}\}$.

\paragraph{\llap{$\blacktriangleright$\ }When $a\leq 0$ and $c\leq 0$,}
then by letting $\alpha'=-\beta$ and $\beta'=-\alpha$, the problem is
transformed into that of maximizing
$\mu'(\alpha',\beta')=\min\{\beta'-\alpha', -a\beta'+b, -c\alpha'+d\}$
over $D'=\{(\alpha',\beta') \mid -M_\alpha\leq \beta'\leq -m_\alpha,\
-M_\beta\leq \alpha'\leq -m_\beta,\ \alpha'\leq \beta'\}$.  Now
$-c\geq 0$ and $-a\geq 0$, and we have reduced this case to the
previous one.

\paragraph{\llap{$\blacktriangleright$\ }When $a\geq 0$ and $c\leq 0$,}
we~split~$\bbR+^2$ into three convex polyhedra, depending on which
of~$h(\alpha,\beta)$, $f(\alpha)$ and $g(\beta)$ is minimal, and look at the position
of~$D$ w.r.t. these polyhedra (see~right of Fig.~\ref{fig-lmtech}):



\begin{itemize}
\item if the upper-right point $(\min\{M_\alpha,M_{\beta}\},M_{\beta})$ is in the
  polyhedron where $f$ is minimal (which can be checked by computing
  the values of $f(\min\{M_\alpha,M_{\beta}\})$, $g(M_{\beta})$, and $h(\min\{M_\alpha,M_{\beta}\},M_{\beta})$),
  then it realizes the maximum over~$D$;
\item similarly, if the lower-left point $(m_\alpha,\max\{m_\alpha,m_\beta\})$ is in the
  polyhedron where $g$ is minimal, then again it realizes the maximum
  over~$D$;
\item similarly, if the upper-left corner $(m_\alpha,M_{\beta})$ is in the polyhedron where
  $h$ is minimal, it realizes the maximum over~$D$;
\item if none of the above cases apply, we consider the tripoint of
  the diagram, whose coordinates are $(\frac{d-b(1-c)}{(a+1)(1-c)-1},
  \frac{d(a+1)-b}{(a+1)(1-c)-1})$. Write $T_\alpha$ and~$T_\beta$ for those coordinates.
  \begin{itemize}
  \item if $T_\beta>M_{\beta}$, then the maximal point is at
    $(\frac{M_{\beta}-b}{a+1},M_{\beta})$;
  \item if $T_\alpha< m_\alpha$, then the maximal point is at
    $(m_\alpha,\frac{m_\alpha+d}{1-c})$;
  \item if $T_\beta<T_\alpha$, i.e. $da<bc$, then $(m_\alpha,m_\beta)$ must be in the
    polyhedron where $f$ is minimal among the three functions, and
    $(M_\alpha,M_{\beta})$ must be in that where $g$ is minimal. We then are in one
    of the following three situations:
    \begin{itemize}
    \item if $(\min(m_\beta,M_\alpha),m_\beta)$ is in the polyhedron where $g$ is
      minimal, then the maximal point is at $(\frac{cm_\beta+d-b}{a}, m_\beta)$;
    \item otherwise, if $(M_\alpha,\max(m_\beta,M_\alpha))$ is in the polyhedron where
      $g$ is minimal, then the maximal point is at
      $(\frac{b-d}{c-a},\frac{b-d}{c-a})$;
    \item otherwise, the maximal point is at $(M_\alpha,\frac{aM_\alpha+b-d}c)$;
    \end{itemize}
  \item finally, if $T_\beta\geq T_\alpha$, we~again have three different cases:
    \begin{itemize}
    \item if $T_\alpha>M_\alpha$, then the maximal point $f(M_\alpha)$
      is at any point between $(M_\alpha,(a+1)M_\alpha+b)$ and $(M_\alpha,\frac{aM_\alpha+b-d}c)$;
    \item if $T_\beta<m_\beta$, then the maximal point $g(m_\beta)$ is reached at any point
      between $(\frac{cm_\beta+d-b}a,m_\beta)$ and $(m_\beta(1-c)-d,m_\beta)$;
    \item otherwise, the maximal point $\frac{da-bc}{(a+1)(1-c)-1}$
      is reached at $(T_\alpha,T_\beta)$.
    \end{itemize}
  \end{itemize}
\end{itemize}
\end{proof}

Table~\ref{table-lemmatech} summarizes the various cases for computing
the maximum value of function
$(\alpha,\beta)\mapsto \min\{\beta-\alpha,a\alpha+b,c\beta+d\}$ over
$D=\{(\alpha,\beta)\in\bbR+^2\mid m_\alpha\leq \alpha\leq M_\alpha,\
m_\beta\leq \beta\leq M_\beta, \alpha\leq \beta\}$.

\begin{table}
\def\arraystretch{1.5}
\noindent\hskip-.1\linewidth
\begin{minipage}{1.2\linewidth}
\paragraph{when $a\leq 0$ and $c\geq 0$: }
\[
\begin{array}{|c|c|}
\hline
\quad\text{coordinates of maximal point}\quad &
\quad\text{value of maximal point}\quad \\
\hline
(m_\alpha, M_\beta) &  \min\{M_\beta-m_\alpha, am_\alpha+b, cM_\beta+d\} \\
\hline
\end{array}
\]

\paragraph{when $a\geq 0$ and $c\geq 0$: }
\[
\begin{array}{|c|c|c|}
\hline
\text{Condition} &
\quad\text{coordinates of maximal point}\quad\null &
\quad\text{value of maximal point}\quad\null \\
\hline
\frac{M_\beta-b}{a+1}\leq m_\alpha
  & (m_\alpha, M_\beta) &  \min\{M_\beta-m_\alpha,cM_\beta+d\} \\
\hline
m_\alpha\leq \frac{M_\beta-b}{a+1}\leq \min\{M_\alpha,M_\beta\}
  & (\frac{M_\beta-b}{a+1}, M_\beta) &  \min\{\frac{aM_\beta+b}{a+1}, cM_\beta+d\} \\
\hline
\min\{M_\alpha,M_\beta\} \leq \frac{M_\beta-b}{a+1}
  & (\min\{M_\alpha,M_\beta\}, M_\beta) &  \min\{aM_\alpha+b,aM_\beta+b,cM_\beta+d\} \\
\hline
\end{array}
\]

\paragraph{when $a\leq 0$ and $c\leq 0$: }
\[
\begin{array}{|c|c|c|}
\hline
\text{Condition} &
\quad\text{coordinates of maximal point}\quad\null &
\quad\text{value of maximal point}\quad\null \\
\hline
M_\beta\leq \frac{m_\alpha+d}{1-c}
  & (m_\alpha, M_\beta) &  \min\{M_\beta-m_\alpha,am_\alpha+b\} \\
\hline
\max\{m_\alpha,m_\beta\}\leq \frac{m_\alpha+d}{1-c}\leq M_\beta
  & (m_\alpha, \frac{m_\alpha+d}{1-c}) &  \min\{\frac{cm_\alpha+d}{1-c}, am_\alpha+b\} \\
\hline
 \frac{m_\alpha+d}{1-c} \leq \max\{m_\alpha,m_\beta\}
  & (m_\alpha, \max\{m_\alpha,m_\beta\}) &  \min\{am_\alpha+b, cm_\beta+d, cm_\alpha+d\} \\
\hline
\end{array}
\]

\paragraph{when $a\geq 0$ and $c\leq 0$: }
\def\mctwo#1&{\multicolumn{2}{|c|}{#1}&}
\[
\begin{array}{|c|c|c|c|}
\hline
\mctwo \text{Condition} &
\quad\text{coordinates of maximal point}\quad\null &
\quad\text{value of maximal point}\quad\null \\
\hline
\mctwo f\leq g, h\text{ at }(\min\{M_\alpha,M_\beta\},M_\beta) & (\min\{M_\alpha,M_\beta\},M_\beta) & \min\{aM_\alpha+b,aM_\beta+b\}) \\ 
\hline
\mctwo g\leq f,h\text{ at }(m_\alpha,\max\{m_\alpha,m_\beta\}) & (m_\alpha,\max\{m_\alpha,m_\beta\}) & \min\{cm_\alpha+d,cm_\beta+d\}
\\\hline
\mctwo h\leq f,g\text{ at }(m_\alpha,M_\beta) & (m_\alpha,M_\beta) & M_\beta-m_\alpha
\\\hline
\multicolumn{4}{l}{\text{If none of the above conditions hold: let }
  T_\alpha=\frac{d-b(1-c)}{(a+1)(1-c)-1}
  \text{ and }
  T_\beta=\frac{d(a+1)-b}{(a+1)(1-c)-1}}
\\\hline
\mctwo T_\beta\geq M_\beta & (\frac{M_\beta-b}{a+1},M_\beta) & \frac{aM_\beta+b}{a+1}
\\\hline
\mctwo T_\alpha\leq m_\alpha & (m_\alpha,\frac{m_\alpha+d}{1-c}) & \frac{cm_\alpha+d}{1-c}
\\\hline
 & g\leq f,h\text{ at }(\min\{m_\beta,M_\alpha\},m_\beta) & (\frac{cm_\beta+d-b}{a},m_\beta) & cm_\beta+d 
\\\cline{2-4}
ad\leq bc & g\leq f,h\text{ at }(M_\alpha,\max\{m_\beta,M_\alpha\}) & (\frac{d-b}{a-c},\frac{d-b}{a-c}) & \frac{ad-bc}{a-c} 
\\\cline{2-4}
& \text{otherwise} & (M_\alpha,\frac{aM_\alpha+b-d}{c}) & aM_\alpha+b
\\\hline
& T_\beta\leq m_\beta & ((1-c)m_\beta-d,m_\beta) & cm_\beta+d
\\\cline{2-4}
ad\geq bc & T_\alpha\geq M_\alpha & (M_\alpha,(a+1)M_\alpha+b) & aM_\alpha+b
\\\cline{2-4}
& \text{otherwise} & (T_\alpha,T_\beta) & \frac{ad-bc}{(a+1)(1-c)-1}
\\\hline
\end{array}
\]
\end{minipage}
\medskip
\caption{Solutions to the optimization problem of Lemma~\ref{lemma-technique}}
\label{table-lemmatech}
\end{table}

\end{document}